\documentclass[submission,copyright,creativecommons]{eptcs}
\providecommand{\event}{AiML 2026} 

\usepackage{aiml26}
\usepackage{dsfont}
\usepackage{bm}
\usepackage{iftex}
\usepackage{stmaryrd}
\usepackage{amssymb}
\usepackage{comment}
\usepackage{mathtools}
\usepackage{stmaryrd}

\usepackage{amsmath}
\usepackage{amsthm}
\usepackage{graphicx}
\usepackage{tikz}

\ifpdf
  \usepackage{underscore}         
  \usepackage[T1]{fontenc}        
\else
  \usepackage{breakurl}           
\fi

\usepackage{amsfonts, amsmath, amssymb}
\usepackage{color}
\usepackage{tikz,url}
\usetikzlibrary{calc,shapes}
\usepackage{hyperref}
\usepackage{amsthm}

\usepackage{comment}

\newcommand{\eq}{\leftrightarrow}
\newcommand{\Eq}{\Leftrightarrow}
\newcommand{\imp}{\rightarrow}
\newcommand{\Imp}{\Rightarrow}

\newcommand{\et}{\wedge}
\newcommand{\vel}{\vee}
\newcommand{\Et}{\bigwedge}

\newcommand{\proves}{\models}

\renewcommand{\phi}{\varphi}
\newcommand{\union}{\cup}
\newcommand{\Union}{\bigcup}
\newcommand{\inter}{\cap}
\newcommand{\Inter}{\bigcap}

\newcommand{\power}{\mathcal P}

\newcommand{\M}{\widehat{K}}

\newcommand{\np}{\overline{p}}

\newcommand{\nq}{\overline{q}}

\newcommand{\lang}{\mathcal L}

\newcommand{\weg}[1]{}

\newcommand{\powerset}{\mathcal{P}}

\newcommand{\see}{\mathsf{see}}

\usepackage{bm}

\def\PDL{\mathbf{PDL}}
\def\IPDL{\mathbf{IPDL}}

\title{Resolving Asynchronous Distributed Knowledge}
\author{Philippe Balbiani
\institute{IRIT, CNRS---INPT---UT}
\email{philippe.balbiani@irit.fr}
\and
Hans van Ditmarsch
\institute{IRIT, CNRS---INPT---UT}
\email{hansvanditmarsch@gmail.com}
\and
Clara Lerouvillois
\institute{IRIT, CNRS---INPT---UT}
\institute{IHPST, Paris 1 Panthéon Sorbonne}
\email{clara.lerouvillois@irit.fr}
}

\newcommand{\titlerunning}{Resolving Asynchronous Distributed Knowledge}
\newcommand{\authorrunning}{Philippe Balbiani et al.}
\hypersetup{
  bookmarksnumbered,
  pdftitle    = {\titlerunning},
  pdfauthor   = {\authorrunning},
  pdfsubject  = {EPTCS},               
  pdfkeywords = {keyword1, keyword2} 
}

\begin{document}
\maketitle
\begin{abstract}
There are by now various epistemic modal logics with intersection modalities for distributed knowledge and intersection update modalities for dynamic phenomena like agents sharing (all their) information, agents receiving information from other agents, and full information protocols. One of those is the logic of {\em Resolving Distributed Knowledge}, by {\AA}gotnes and Wang. It has {\em distributed knowledge} modalities for arbitrary subsets of the set of all agents and it also has so-called {\em resolution} modalities for arbitrary subsets of agents sharing their knowledge. In that logic, the agents not involved in the knowledge sharing are aware of the agents sharing knowledge, agents are memory-less, and the kind of dynamics represents synchronous updates, where there is common awareness of the global clock. In contrast, in this contribution we present a logic for {\em Resolving Asynchronous Distributed Knowledge}. It is an asynchronous generalization of the synchronous logic of resolving distributed knowledge. The logical semantics is history-based: truth is not only with respect to a given world in a model, but also with respect to a given history of prior resolutions, of which each individual agent can only observe a part. In particular, an agent is unaware of resolutions for groups of agents not including her. As is to be expected, this comes with many technical complications, for example concerning the axiomatization. The synchronous axioms relating resolution to distributed knowledge are now invalid. The modelling advantages of such an asynchronous novel logic, for distributed computing and similar areas, are however substantial and a major asset.
  \end{abstract}

\section{Introduction} \label{sec.introduction}

If Anne ($a$) knows $p$ and Bill ($b$) knows $p \imp q$ then neither knows $q$ but they still have {\em distributed knowledge} of $q$. If they were to {\em share} their information, they would both know $q$. Distributed knowledge modalities are {\em intersection modalities}: $p$ is distributed knowledge for $a$ and $b$ in a world $w$, iff $p$ is true in all worlds $v$ which are indistinguishable for $a$ {\em and} for $b$ from the actual world $w$. In this case, distributed knowledge of $p$ becomes common knowledge of $p$ by sharing. However, as is well known, propositions may be distributed knowledge even when the agents cannot share this information. For another scenario, if $p$ is true but $a$ does not know this, and $b$ knows that, then $a$ and $b$ have distributed knowledge of the former, as it is true in all worlds which are indistinguishable for $a$ and for $b$. But when $b$ informs $a$ of $p$, or even of $p$ and that she does not know $p$, agent $a$ learns that $p$ which makes her ignorance false. Modalities for sharing information are {\em intersection updates} (intersection update modalities): in a world $w$ proposition $q$ is true after $a$ and $b$ share their knowledge iff $q$ is true in the model wherein we have replaced the indistinguishability relations for $a$ and $b$ by the intersection of these two relations. Let now Cath ($c$) enter the scene and let us replay the first of the above scenarios, where all three agents $a,b,c$ are initially aware that $a$ knows $p$ and $b$ knows $p \imp q$. Again Anne and Bill share their knowledge. What does Cath learn from this? That depends. If we assume synchrony, then Cath learns that Anne and Bill share their knowledge, so afterwards Cath knows that Anne and Bill now both know $q$. Also, Anne and Bill know that Cath knows this. But if we assume asynchrony, Anne and Bill may have shared their knowledge without Cath learning that. In which case Cath still considers it possible that Anne does not know $q$. But also Anne and Bill now face more uncertainty. For example, if after Anne and Bill share their knowledge, Anne and Cath share their knowledge, Cath now knows that Anne knows $p$. But Bill, who was not involved in the second sharing, does not know that. And so on. Various logics have been proposed to formalize distributed knowledge and sharing knowledge and their interaction, and that should be seen as incorporating synchrony. Instead, we propose an asynchronous logic to formalize distributed knowledge and sharing knowledge and their interaction.

The remainder of this introduction is a succinct survey of modal logics with intersection modalities, modal logics with intersection updates, and modal logical approaches to asynchrony.

{\bf Intersection modalities.}
If we consider a modal language with modalities $\Box_a$ and $\Box_b$, interpreted on Kripke models with binary relations $R_a$ and $R_b$, resp., the problem of intersection modalities boils down to the observation that adding a novel modality (suggestively named) $\Box_{a\cup b}p$ is true in a world $w$ iff $p$ is true in all worlds $v$ which are $R_{a}$-accessible \emph{or} $R_{b}$-accessible does not create problems. This is so because in such a case, $R_{a\cup b}=R_{a}\cup R_{b}$ and one simply has that $\Box_{a\cup b}p$ is equivalent to $\Box_{a}p\wedge\Box_{b}p$.
Whereas adding a modality $\Box_{a\cap b}$ such that $\Box_{a\cap b}p$ is true in a world $w$ iff $p$ is true in all worlds $v$ which are $R_{a}$-accessible \emph{and} $R_{b}$-accessible creates problems. One now has that $R_{a\cap b}=R_{a}\cap R_{b}$ and we can no longer define $\Box_{a\cap b}p$ with the other modalities, the canonical model is not of the right kind.
One way to address this is to expand the logical language with {\em nominals} \cite{Passy:Tinchev:1985,Passy:Tinchev:1991} which led to the development of {\em hybrid logics} \cite{Blackburn:Seligman:1995,Areces:tenCate:2007} and related logics \cite{Goranko:1996,Goranko:Passy:1992,derijke:1992}. Another way to address this led to the development of modal logics with intersection modalities, where we highlight Propositional Dynamic Logic ($\PDL$) with intersection~\cite{Danecki:1985,Harel:1985}, and epistemic logics with distributed knowledge~\cite{FaginHV92,halpernmoses:1990,AlechinaBS12}, although there are many further approaches such as Boolean Modal Logic~\cite{Gargov:Passy:1990,Gargov:et:al:1986}, and knowledge representation logics~\cite{Orlowska:1990,Vakarelov:1991}. 

{\bf Propositional dynamic logic.}
In $\PDL$ we generalize from modalities $\Box_a$ to modalities $\Box_\alpha$ where $\alpha$ is a program and that are interpreted in models with recursively defined relations $R_\alpha$, with the same basic programs $a,b,\dots$ as above but apart from $\union$ operations of sequential execution and arbitrary iteration, as well as another basic program called test \cite{hareletal:2000}. By also adding intersection $\alpha\cap\beta$ of programs to $\PDL$, we obtain $\IPDL$ \cite{Danecki:1985,Harel:1985}. Here, $\alpha\cap\beta$ represents parallel execution of programs $\alpha$ and $\beta$ as $R_{\alpha\cap\beta}=R_\alpha\cap R_\beta$, so again, similar to $R_a \cap R_b$ above. Various works involving the complexity \cite{Lange:2005,Lange:Lutz:2005,Massacci:2001} and axiomatization \cite{balbianietal:2003} of $\IPDL$ have appeared.

{\bf Distributed knowledge.}
The notion of {\em distributed knowledge} is rooted in sociology, economics, and philosophy \cite{HayekAER45,hilpinen:1977,swanson:1986} under different terms; \cite{HayekAER45} is an admirable pamphlet against centralized planning and in favour of distributed planning (and authority), \cite{hilpinen:1977} calls it `impersonal knowledge', and \cite{swanson:1986} `undiscovered public knowledge'. The earliest epistemic logical source is \cite{halpernmoses:85b} (the later journal version \cite{halpernmoses:1990} also gives a logical semantics), wherein the notion was called `implicit knowledge', followed on the heals by \cite{ParikhR85} who proposes an axiomatization (without claiming completeness) in a temporal epistemic logic. In fact their non-temporal fragment is the, yet somewhat later, complete axiomatization of \cite{FaginHV92,HoekM92}, where these publications achieved completeness by different methods. All such works interpret modalies $\Box_a$ not on models with arbitrary relations $R_a$ but on models with equivalence relations. Among the many issues of further interest concerning distributed knowledge, particularly in view of representing multiple agents sharing knowledge given uncertainty about their own and each other's knowledge, is a certain discrepancy between syntactic and semantic intuitions of distributed knowledge, and what sharing knowledge actually means. Such issues were already discussed in the original \cite{FaginHV92,HoekM92} and continue to be investigated in more recent times \cite{rusbouke.aiml:2024}.

{\bf Intersection updates.}
After the history of intersection modalities we now proceed with the history of intersection updates, the history of {\em sharing} information. Given a Kripke model with two (typically) equivalence relations $R_a$ and $R_b$, the {\em intersection update modality} $\Box_{a \inter b}^{\mathsf{update}}$ is such that $\Box_{a \inter b}^{\mathsf{update}} p$ is true in a world $w$ of a model $M$ iff $p$ is true in the same world $w$ but in the \emph{updated} model $M^{a \inter b}$ wherein the relations $R_a$ and $R_b$ are replaced by the relation $R_a \inter R_b$. That is all. Note that when we interpret the intersection modality we (may) {\em change the point of evaluation} in a given model but we {\em do not change the model} wherein we evaluate the formula bound by the modality, whereas if we interpret the intersection update modality we {\em do not change the point of evaluation} but we (may) {\em change the model} wherein we evaluate the formula bound by the modality. Intersection update modalities are therefore interpreted as updates of Kripke models and they can indeed be seen as a further development in {\em dynamic epistemic logics}, such as public announcement logic \cite{plaza:2007} and action model logic \cite{baltagetal:1998} (for references see \cite{hvdetal.del:2007,moss.handbook:2015}), although these are updates with very different properties.

The first paper involving the intersection update to our knowledge was \cite{stefanjohan:2009}, in a mixed setting of epistemic and inquisitive logic. The intersection update here is the {\em resolve} action, encoding the answer to a question thus resolving an issue (intersecting an epistemic relation with an issue relation). On the heals of \cite{stefanjohan:2009} were a number of related publications \cite{BaltagS13,carrington:2013,boddy:2014,goldbach:2015,baltagetal.hintikka:2018}. Here, \cite{BaltagS13} introduces intersection updates on plausibility structures whereas \cite{baltagetal.hintikka:2018} further generalizes the inquisitive and epistemic setting of \cite{stefanjohan:2009}. Independently, slightly later again, \cite{AgotnesW17} focuses on intersection updates (called \emph{resolution}) and distributed knowledge for arbitrary sets of agents and their interaction. Further developments introduce uncertainty even among agents who have shared information \cite{Baltag20,baltagsmets.aiml:2024}, which also relates to a different strand of research coming out of distributed computing, and intersection updates incarnate the execution of {\em communication graphs} or {\em communication patterns} \cite{diego:2021,cdrv:2023,armandoetal.tark:2023,diego:2019,diego:2024}. The operation of pooling or basic intersection in \cite{ChristoffGratzlRoy2022} corresponds to a intersection (update) modality similar to what we call resolution here; this updates models with plausibility relations. An attempt to combine distributed knowledge and intersection updates into a single modality for {\em sharing knowledge} is \cite{BalbianiD24}.

{\bf Asynchrony.} Intersection modalities like distributed knowledge and intersection updates such as resolution and communication patterns assume synchrony (and even so in distributed computing, were a \emph{round} of asynchronous events represents a snapshot recording time). Things become more complex with (full) asynchrony, in the absence of a global clock. This was already addressed in \cite{halpernmoses:1990}, and `full information protocols' in distributed computing \cite{MosesT88} correspond to intersection updates. In dynamic epistemic logics such asynchrony requires a history-based semantics \cite{parikhetal:2003,degremontetal:2011}. Even with more restricted information exchange such as in gossip protocols this leads to far more challenging axiomatizations (and higher complexities) \cite{logicofgossiping:2020,hvdetal.lucky:2024,AptKW17}.

{\bf Overview of content.} In Section~\ref{sec.rdk} we review the logic of resolving (synchronous) distributed knowledge. In Section~\ref{sec.radk} we introduce the logic of resolving asynchronous distributed knowledge, and compare it to the logic of synchronous distributed knowledge.  In Section~\ref{sec.axiomatization} we define an infinitary axiomatization for the logic of resolving asynchronous distributed knowledge and prove its completeness.  Section~\ref{sec.further} lists further research, on redundant resolutions, derivable rules, expressivity, and common knowledge.

\section{Resolving distributed knowledge}   \label{sec.rdk}

We briefly present {\AA}gotnes and Wang's logic for resolving distributed knowledge \cite{AgotnesW17}. Let $A$ be a finite non-empty set of \emph{agents}, and $P$ a countable set of propositional variables (atoms). The language $\lang_{DR}$ is defined by $\phi ::= p \mid \top \mid (\phi \et \phi) \mid \neg \phi \mid D_B \phi \mid R_B \phi$, where $p \in P$ and $B \subseteq A$.  The sublanguage $\lang_{D}$ with only modalities $D_B\phi$ is the language of distributed knowledge. For $\Box_a$ we now write $K_a$, and it is defined by abbreviation as $D_{\lbrace a \rbrace}$. Formula $R_B \phi$ is read as `after resolution for group of agents $B$, $\phi$ (is true)'.

The structures are multi-agent {\em epistemic models} $(W,\sim,V)$ for the set $A$ of agents (where $(W,\sim)$ is a multi-agent \emph{frame}). We can see $\sim$ as a set $\{\sim_a\}_{a \in A}$ of binary {\em indistinguishability relations} (or as a function mapping each agent to such a relation $\sim_a$). 
Valuation $V$ is a function from the set of atoms to the powerset of $W$ mapping each atom to the subset of worlds where it is true. We write ${\sim_B} := {\bigcap_{a \in B} \sim_a}$ for the epistemic relation of a group $B$ (so that $\sim_\emptyset \ = W \times W$).

If $M = (W,\sim,V)$ and $B\subseteq A$ then {\em updated model} $M^B := (W,\sim^B,V)$ is the result of resolution for group $B$, where ${\sim^B_a} := {\bigcap_{b \in B} \sim_b}$ if $a \in B$ and ${\sim^B_a} := {\sim_a}$ otherwise. In updated model $M^B$ the relations for the agents $a \in B$ have been updated from $\sim_a$ to $\sim_B$. Hence, in $M^B$ we have that ${\sim^B_B} = {\sim^B_a}$ for all $a \in B$, unlike in $M$. Given these notations, we have that ${\sim^B_C} = {\sim_C}$ if $B \cap C = \emptyset$ and that ${\sim^B_C} = {\sim_{B\cup C}}$ if $B\cap C \neq \emptyset$. Note that $M^a = M$ and that $M^\emptyset = M$.  For $M^{\{a_1,\dots,a_n\}}$ we write $M^{a_1\dots a_n}$, and for $(M^B)^C$ we write $M^{B.C}$.

Let a vector $\vec{G}$ represent a sequence $B_1\dots B_n$, where $B_i \subseteq A$ for $1 \leq i \leq n$---the empty sequence is $\epsilon$. 
By $|\vec{G}|$ we denote the length of $\vec{G}$, and $a \notin \vec{G}$ means that $a\notin B$ for any $B \subseteq A$ such that $B$ occurs in $\vec{G}$. 
We write $\sim^{\vec{G}}_a$ for the accessibility relation of an agent $a$ in $M^{\vec{G}}$.

We now present the semantics. The satisfaction relation $\models$ is defined by induction on $\phi \in \lang_{DR}$, where $p\in P$ and $G \subseteq A$.
\[\begin{array}{lcccl}
M,w \models p & &\text{iff} & &w \in V(p) \\
M,w \models \top & &\text{iff}  & &\text{true} \\
M,w \models \neg \phi & &\text{iff} & &M,w \not\models\phi \\
M,w \models \phi\et\psi & &\text{iff} & &M,w \models\phi \text{ and } M,w \models \psi \\
M,w \models D_B \phi & &\text{iff} & &M,v \models \phi \text{ for all } v \in W \text{ such that } w \sim_B v \\
M,w \models R_B \phi & &\text{iff} & &M^B,w \models \phi 
\end{array}\]
A formula $\phi\in\lang_{DR}$ is \emph{valid} iff for all models $M = (W,\sim,V)$ and for all $w \in W$, $M,w \models\phi$. 

    \begin{table}
    \centering
    \begin{tabular}{lll}
         (taut)  & all instantiations of propositional tautologies & \\
         (K$_D$) & $D_B(\phi \imp \psi)\imp (D_B\phi \imp D_B \psi)$ & \\
         (T$_D$) & $D_B \phi \imp \phi$ & \\
         (5$_D$) & $\lnot D_B\phi \imp D_B \lnot D_B \phi$ & \\
         (DG)    & $D_B \phi \imp D_C \phi$ & if $B \subseteq C$ \\
         (RA)    & $R_B p \eq p$ & \\
         (RN)    & $R_B \lnot \phi \eq \lnot R_B \phi$ & \\
         (RC)    & $R_B (\phi \land \psi) \eq (R_B \phi \land R_B \psi)$ & \\
         (RD1)   & $R_B D_C \phi \eq D_{B \cup C}R_B\phi$ & when $B\cap C \neq \emptyset$\\
         (RD2)   & $R_B D_C \phi \eq D_C R_B\phi$ & when $B\cap C = \emptyset$ \\
         (NecD)  & From $\phi$, infer $D_B \phi$ & \\
         (NecR)  & From $\phi$, infer $R_B \phi$ & \\
         (MP)    & From $\phi$ and $\phi \imp \psi$, infer $\psi$ &
    \end{tabular}
    \caption{Axiomatisation \textbf{RD} for the logic of resolving distributed knowledge}
    \label{axiomatisationRD}
\end{table}

The axiomatization {\bf RD} for the logic of resolving distributed knowledge in Table~\ref{axiomatisationRD} extends that of the logic of distributed knowledge with reduction axioms for resolution and a derivation rule for necessitation of resolution. 
Completeness of this axiomatization is shown by reducing $\lang_{DR}$-formulas to $\lang_{D}$-formulas, by which we mean that any formula with resolution and distributed knowledge modalities is provably (and semantically) equivalent to a formula without resolution modalities. It can be shown that {\em Replacement of Equivalents} (RE: from $\phi \eq \psi$ derive $\chi[p/\phi] \eq \chi[p/\psi]$, where $\chi[p/\phi]$ is uniform substitution of $p$ in $\chi$ by $\phi$) is derivable in \textbf{RD}. This is needed to reduce formulas of form $R_G R_H \phi$ where $G \neq H$, as there is no axiom of shape $R_G R_H \phi \eq \dots$

\section{Resolving asynchronous distributed knowledge} \label{sec.radk}


We now propose a logical framework for resolving {\em asynchronous} distributed knowledge. The logical language $\lang_{DR}$ is the same, whereas accessibility relations are now encoding asynchronous distributed knowledge. Given asynchrony, we propose a history-based semantics. That is, instead of interpreting formulas $\phi$ in pointed models $(M,w)$, where such an $M$ could be an updated $M^{\vec{G}}$, we now wish to interpret formulas in pointed models $(M,w,\vec{G})$, where the resolution sequence explicitly remains at our disposition. This is necessary, because what an agent knows now may be different when she was involved in the last resolution in $\vec{G}$ than from when she was not. In order to compare pairs $(w,\vec{G})$ and $(v,\vec{H})$ we not only need the indistinguishability relation $\sim_a$ between worlds $w$ and $v$ but also a novel \emph{resolution relation} $\approx_a$ between resolution sequences $\vec{G}$ and $\vec{H}$. Connecting both relations is the \emph{view} by agent $a$ of sequence $\vec{G}$, denoted $\see_a(\vec{G})$, that is the set of agents $B$ with whom $a$ has shared the equivalence relations (e.g., recalling the introduction, after resolution $ab$, she has shared it with $b$ as the new $\sim_a$ is now $\sim_{a \inter b}$). The set $\see_a(\vec{G})$ determines the knowledge gained by agent $a$ after history $\vec{G}$---or after any other history $\vec{H}$ that she cannot distinguish from $\vec{G}$.
Let us begin by defining the resolution relation and the view, and show some relevant results relating them.

{\bf Resolution relation.} 
Let $a \in A$, and $\vec{G},\vec{H} \in \power(A)^\ast$ be given. The {\em resolution relation} $\approx_a$ between resolution sequences is the equivalence closure of:
\[ \begin{array}{lcll}
    \epsilon \approx_a \epsilon & \text{iff} & \text{true} \\
    \vec{G}.B \approx_a \vec{H} & \text{iff} & \vec{G} \approx_a \vec{H} &\text{when } a \notin B \\
    \vec{G}.B \approx_a \vec{H}.B & \text{iff} & \vec{G} \approx_b \vec{H} \text{ for all } b \in B \hfill \qquad\qquad\qquad &\text{when } a \in B
\end{array} \]
We further define ${\approx_B} := {\Inter_{a\in B} \approx_a}$. Note that $\approx_\emptyset$ is the universal relation.


{\bf View.} 
The \emph{view} $\see_a(\vec{G})$ of agent $a$ of a resolution sequence $\vec{G}$ is the set $B \subseteq A$ of agents such that given any model $M = (W,\sim,V)$, agent $a$'s current relation $\sim_a^{\vec{G}}$ is $\inter_{b \in B} \sim_b$. Let $B \subseteq A$, then: 
\begin{itemize}
\item $\see_B(\epsilon) := B$; 
\item $\see_B(\vec{G}.C) := \see_{B \union C}(\vec{G})$ when $B \inter C \neq \emptyset$;
\item $\see_B(\vec{G}.C) := \see_B(\vec{G})$ when $B \inter C = \emptyset$. 
\end{itemize}
For $\see_{\{a\}}(\vec{G})$ we write $\see_a(\vec{G})$. Note that $\see_B(\vec{G}) = C$ does not mean that the view of all agents in $B$ is $C$, but only that the union of the views of all agents in $B$ is $C$. It is easy to see that $\see_B(\vec{G}) = \Union_{b \in B} \see_b(\vec{G})$ (induction on the length of $\vec{G}$).
%
%
\medskip

Before proceeding with the definition of the semantics, we state some properties of the resolution relation and the view, that will later prove useful. Straightforward proofs have been omitted.

\begin{lemma}\label{lemmaApproxDelete}\label{coApproxDeleteGroup}
 If $a\notin \vec{I}$, then $\vec{G}\approx_a \vec{H}.\vec{I}$ if, and only if, $\vec{G} \approx_a \vec{H}$.
\end{lemma}

Similarly, if $\vec{I}\in \power(A\setminus B)^\ast$, then $\vec{G}\approx_B \vec{H}.\vec{I}$ if, and only if, $\vec{G} \approx_B \vec{H}$.

\begin{lemma}\label{lemmaApproxEmpty} \label{corolApproxEmptyIntersection}
$\epsilon \approx_B \vec{G}$ if, and only if, $\vec{G} \in \power(A\setminus B)^\ast$.
\end{lemma}

Consequently, we also have as a corollary that whenever $B\cap C = \emptyset$, then $C \approx_B \vec{G}$ iff $\vec{G}\in \power(A\setminus B)^\ast$.

\begin{lemma}\label{lemmaDecompositionSingleAgent}
If $a \in B$, then $\vec{G}.B \approx_a \vec{H}$ iff there are $\vec{H_1}$, $\vec{H_2}\in \power(A)^\ast$ such that $\vec{H} = \vec{H_1}.B.\vec{H_2}$ with $\vec{G}\approx_B \vec{H_1}$ and $a \notin \vec{H_2}$.	
\end{lemma}

\weg{
\begin{proof}
     The proof proceeds by induction on $\vec{H}$. If $\vec{H}=\epsilon$, since $a\in B$, $\vec{G}.B\not\approx_a \vec{H}$. If $\vec{H}=\vec{H'}.C$, we distinguish case $a\notin C$ from case $a\in C$. 
If $a\notin C$, then $\vec{G}.B\approx_a \vec{H'}.C$ implies that $\vec{G}.B \approx_a \vec{H'}$, so, by inductive hypothesis, there are $\vec{H'_1}, \vec{H'_2}$ such that $\vec{H'}=\vec{H_1'}.B.\vec{H'_2}$, $\vec{H'_1}\approx_B \vec{G}$ and $a\notin \vec{H'_2}$. We now take $\vec{H_1}:=\vec{H'_1}$ and $\vec{H_2}:=\vec{H'_2}.C$ to obtain the desired result.
If $a\in C$, then $\vec{G}.B\approx_a \vec{H'}.C$ implies that $B =C$ and $\vec{G} \approx_B \vec{H'}$, so we can simply take $\vec{H_1}:= \vec{H'}$ and $\vec{H_2}:=\epsilon$.
\end{proof}
}

\begin{lemma}\label{lemmaApproxNonEmptyIntersection}
If $B\cap C \neq \emptyset$, then $C \approx_B \vec{H}$ iff there are $\vec{H_1}$, $\vec{H_2}\in \power(A)^\ast$ such that $\vec{H} = \vec{H_1}.C.\vec{H_2}$ with $\vec{H_1} \in \power(A{\setminus}(B\cup C))^*$ and $\vec{H_2}\in \power(A\setminus B)^\ast$.	
\end{lemma}

\begin{lemma}\label{lemmaCompositionSequences}
    If $\vec{G}\approx_{see_B(\vec{H})}\vec{I}$ and $\vec{H}\approx_B \vec{J}$, then $\vec{G}.\vec{H}\approx_B \vec{I}.\vec{J}$.
\end{lemma}

\begin{proof} 
    The proof proceeds by induction on the length of $\vec{H}$. If $\vec{H}=\epsilon$, suppose $\vec{G}\approx_B \vec I$ (because $see_B(\epsilon)=B$) and $\epsilon \approx_B \vec{J}$. by Lemma \ref{coApproxDeleteGroup}, $\epsilon \approx_B \vec{J} \Imp \vec{J} \in \power(A \setminus B)$ and then $\vec{G}\approx_B \vec{I}$ implies $\epsilon.\vec{G}=\vec{G}\approx_B \vec{I}.\vec{J}$. If $\vec{H}=\vec{H'}.C$, suppose $\vec{G}\approx_{see_B(\vec{H'}.C)}\vec{I}$ and $\vec{H'}.C\approx_B \vec{J}$. We now distinguish case $B\cap C = \emptyset$ from case $B\cap C \neq \emptyset$.
    \begin{itemize}
        \item If $B\cap C = \emptyset$, $see_B(\vec{H'}.C)=see_B(\vec{H'})$ and $\vec{H'}.C\approx_B \vec{J} \Imp \vec{H'}\approx_B \vec{J}$ (Corollary \ref{coApproxDeleteGroup}) so, by inductive hypothesis, $\vec{G}.\vec{H'}\approx_B \vec{I}.\vec{J}$. Hence $\vec{G}.\vec{H'}.C\approx_B \vec{I}.\vec{J}$ by Corollary \ref{coApproxDeleteGroup} again.
        \item If $B \cap C \neq \emptyset$, then $see_B(\vec{H'}.C)=see_{B\cup C}(\vec{H'})$. Note that since $B \subseteq B\cup C$, $\vec{G}\approx_{see_{B\cup C}(\vec{H'})}\vec{I}$ implies $\vec{G}\approx_{see_B(\vec{H'})}\vec{I}$ (and respectively for $C)$. Let $a\in B \cap C$. By Lemma \ref{lemmaDecompositionSingleAgent}, $\vec{H'}.C\approx_a \vec{J}$ implies $\vec{J}=\vec{J_1}.C.\vec{J_2}$ where $\vec{H'}\approx_C \vec{J_1}$ and $a\notin \vec{J_2}$. Then, we have $\vec{G}\approx_{see_C(\vec{H'})}\vec{I}$ and $\vec{H'}\approx_C \vec{J_1}$ so, by induction hypothesis, $\vec{G}.\vec{H'}\approx_C \vec{I}.\vec{J_1}$. Hence $\vec{G}.\vec{H'}.C\approx_a \vec{I}.\vec{J_1}.C.\vec{J_2}=\vec{I}.\vec{J}$. Let now $a\in B \setminus C$. Then, $\vec{H'}.C\approx_a \vec{J}$ iff $\vec{H'} \approx_a \vec{J}$, and we have $\vec{G}\approx_{see_a(\vec{H'})}\vec{I}$ (for $see_a(\vec{H'}.C)=see_a(\vec{H'})$). So, by induction hypothesis, $\vec{G}.\vec{H'}\approx_a \vec{I}.\vec{J}$ and hence $\vec{G}.\vec{H'}.C\approx_a \vec{I}.\vec{J}$. Therefore $\vec{G}.\vec{H'}.C\approx_B \vec{I}.\vec{J}$
    \end{itemize}
\end{proof}

\begin{lemma} \label{lemmaApproxToSee} \label{coApproxToSee} \label{lemmaApproxToSim} \label{coApproxToSim}
$\vec{G} \approx_a \vec{H}$ implies $\see_a(\vec{G}) = \see_a(\vec{H})$.
\end{lemma}
\begin{proof}
The proof is by induction on the length of $\vec{G}$.

Base case. We have:
$\epsilon\approx_a\vec{H}$ implies $a \notin \vec{H}$ (Lemma \ref{lemmaApproxEmpty}) so $\see_a(\epsilon) = \{a\} = \see_a(\vec{H})$.

Induction case. Now consider $\vec{G}.B$. We distinguish $a \notin B$ from $a \in B$.
If $a \notin B$ we have:
$\vec{G}.B\approx_a\vec{H}$ iff (by definition) $\vec{G}\approx_a\vec{H}$ , which implies (induction) $\see_a(\vec{G}) = \see_a(\vec{H})$, so $\see_a(\vec{G}.B) = \see_a(\vec{H})$.
If now $a \in B$ the resolution sequence compared with must have shape $\vec{H}.B.\vec{I}$ where $a\notin \vec{I}$ (Lemma \ref{lemmaDecompositionSingleAgent}) so that:
$\vec{G}.B\approx_a\vec{H}.B.\vec{I}$ iff $\vec{G}.B\approx_a\vec{H}.B$ (Lemma \ref{lemmaApproxDelete}), which implies $\vec{G}\approx_b\vec{H}$ for all $b \in B$. This implies (induction) $\see_b(\vec{G}) = \see_b(\vec{H})$ for all $b \in B$, which implies $\Union_{b \in B}\see_b(\vec{G}) = \Union_{b \in B}\see_b(\vec{H})$ so (by definition) $\see_a(\vec{G}.B) = \see_a(\vec{H}.B) = \see_a(\vec{H}.B.\vec{I})$.
\end{proof}
Consequently, we also have that $\vec{G} \approx_B \vec{H}$ implies $\see_B(\vec{G}) = \see_B(\vec{H})$. From \cite{BalbianiD24} we further recall that for all models $M=(S,\sim,V)$, ${\sim^{\vec{G}}_B} = {\sim_{\see_B(\vec{G})}}$. In other words, we might as well have stated that $\vec{G} \approx_B \vec{H}$ implies ${\sim^{\vec{G}}_B} = {\sim^{\vec{H}}_B}$, or that $\vec{G} \approx_B \vec{H}$ implies $\sim_{\see_B(\vec{G})} = \sim_{\see_B(\vec{H})}$, and will quote Lemma~\ref{lemmaApproxToSee} in such cases.

However, $\see_a(\vec{G}) = \see_a(\vec{H})$ does not imply $\vec{G} \approx_a \vec{H}$. Typical counterexamples are that $\see_a(\vec{G}.a)=\see_a(\vec{G})$ whereas $\vec{G}.a \not\approx_a \vec{G}$ and $\see_a(ab.ab) = \see_a(ab)$ whereas $ab.ab \not\approx_a ab$.

We proceed to define the semantics.

\begin{definition}[Semantics]\label{defSemantics}
By induction on $\phi \in \lang_{DR}$, where $p\in P$, $B \subseteq A$ and $\vec{G} \in \powerset(A)^*$. 
\[ \begin{array}{lll}
M,w,\vec{G} \models p & \text{iff} & w \in V(p) \\
M,w,\vec{G} \models \top & \text{iff} &\text{true} \\
M,w,\vec{G} \models \neg \phi & \text{iff} & M,w,\vec{G} \not\models\phi \\
M,w,\vec{G} \models \phi\et\psi & \text{iff} & M,w,\vec{G} \models\phi \text{ and } M,w,\vec{G} \models \psi \\
M,w,\vec{G} \models D_B \phi & \text{iff} & M,v,\vec{H} \models \phi \text{ for all } v \in W, \ \vec{H} \in \power(A)^\ast \text{ such that } w \sim^{\vec{G}}_B v \text{ and } \vec{G} \approx_B \vec{H} \\
M,w,\vec{G} \models R_B \phi & \text{iff} & M,w,\vec{G}.B \models \phi 
\end{array} \]
\end{definition}
As not uncommon in history-based semantics, two notions of validity emerge. A formula $\phi\in\lang_{DR}$ is \emph{$\epsilon$-valid} if $M,w,\epsilon \models\phi$ for all models $M = (W,\sim,V)$ and for all $w \in W$. A formula $\phi\in\lang_{DR}$ is \emph{$\ast$-valid} (or {\em always valid}) if $M,w,\vec{G} \models\phi$ for all $M = (W,\sim,V)$, $w \in W$, and all $\vec{G} \in \powerset(A)^*$. 
%
We should acknowledge that we are uncertain if $\epsilon$-validity and $\ast$-validity correspond. This is a question left for future research.

This asynchronous semantics looks very much like the synchronous semantics of the previous section, except for the clause for distributed knowledge. Let us be explicit about some differences and correspondences. First, Example~\ref{example.asdf} shows that $M,w,\vec{G} \models \phi$ is not equivalent to (in the synchronous semantics) $M^{\vec{G}},w \models \phi$. So there is a real difference.

\begin{example} \label{example.asdf}
Consider three agents $a,b,c$ and model $M$ and updated $M^{ab}$ as in Figure \ref{fig1}. Instead of naming worlds by $w$, $v$, etcetera, we name them with their valuation, where $\np$ denotes $\lnot p$. Reflexive arrows are omitted. We also assume symmetry and transitivity.


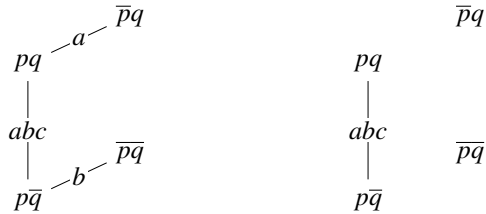
\begin{figure}[h]
\begin{center}
\scalebox{.9}{
\begin{tikzpicture}
    \node (10) at (0,0)     {$p\nq$};
    \node (00) at (1.5,0.7) {$\np\nq$};
    \node (11) at (0,2)     {$p q$};
    \node (01) at (1.5,2.7) {$\np q$};
    \node (10b) at (5,0)     {$p\nq$};
    \node (00b) at (6.5,0.7) {$\np\nq$};
    \node (11b) at (5,2)     {$p q$};
    \node (01b) at (6.5,2.7) {$\np q$};
    \draw (10) -- node[fill=white, inner sep= 1pt] {$b$} (00);
    \draw (10) -- node[fill=white, inner sep= 1pt] {$abc$} (11);
    \draw (11) -- node[fill=white, inner sep= 1pt] {$a$} (01);
%
    \draw (10b) -- node[fill=white, inner sep= 1pt] {$abc$} (11b);
\end{tikzpicture}
}
\end{center}
\caption{An epistemic model $M$ (on the left) and its update $M^{ab}$ (on the right) with three agents.}
\end{figure}\label{fig1}


\noindent
We now have that $M^{ab},pq,\epsilon \models K_cK_a p$ but not that $M,pq,ab \not\models K_cK_a p$. Agent $c$ is unaware of agents $a$ and $b$ resolving their knowledge in model $M$: resolution $ab$ is indistinguishable from the empty sequence $\epsilon$ for her. She therefore does not know agent $a$ has learnt the truth about $p$ as a consequence of this resolution. Since $M,p q, \epsilon \not\models K_a p$ and $ab \approx_c \epsilon$, therefore $M,pq,ab \not\models K_cK_a p$. And therefore also $M,pq,\epsilon \not\models R_{ab}K_cK_a p$. Now consider the prior synchronous semantics. Then $M^{ab},pq \models K_cK_a p$ and therefore  $M,pq \models R_{ab}K_cK_a p$.
\end{example}

But we can also look at this in another way, rather illustrating a correspondence. Let us for a vanishing moment consider a {\em synchronous} distributed knowledge modality $\bm{D}_B \phi$, interpreted as follows on our history-based models, wherein we only have replaced $\vec{G} \approx_B \vec{H}$ by $\vec{G} = \vec{H}$.
\[ \begin{array}{lll}
M,w,\vec{G} \models \bm{D}_B \phi & \text{iff} & M,v,\vec{H} \models \phi \text{ for all } v \in W, \ \vec{H} \in \power(A)^\ast \text{ such that } w \sim^{\vec{G}}_B v \text{ and } \vec{G} = \vec{H}
\end{array} \]
Now write $M,w,\vec{G}\models \phi$ anywhere for the synchronous $M^{\vec{G}},w\models\phi$ (while replacing all $D_B$ by $\bm{D}_B$). This embeds the synchronous semantics into the asynchronous semantics. So, with respect to the above example, indeed, $M,pq,ab \not\models K_cK_a p$, but on the other hand we now have $M,pq,ab \models \bm{K}_c\bm{K}_a p$, which after all corresponds to $M^{ab}, pq \models K_cK_a p$. Isn't that neat? 

However, let us now go back to one language and two different semantics again. A further maybe somewhat curious observation is that resolution $R_B$ has the same semantics either way, only the interpretation of distributed knowledge $D_B$ is different synchronously and asynchronously. Despite the identical semantics, with asynchrony, resolution $R_B$ encodes {\em partial synchronization} for group of agents $B$ \emph{without} any agents not in $B$ being aware of that, whereas, with synchrony, resolution $R_B$ encodes {\em full synchronization} for all agents however with aspects of {\em partial observation}: group of agents $B$ jointly learn (all) each other's knowledge whereas all agents not in $B$ learn that, but not what the agents in $B$ learn. The agents not in $B$ only partially observed the resolution. 

Preparing the ground for the asynchronous axiomatization presented in the next section, let us review what parts of the synchronous axiomatization remain valid and what are now invalid. It is fairly simple. All axioms and rules of {\bf RD} remain valid (or validity preserving), except (RD1) and (RD2); and maybe (NecR). If $\epsilon$-valid and $*$-valid were the same (the open question), then we would have necessitation of resolution. (It is easy to see why: assume $\proves \phi$ and (NecR). Given arbitrary $(M,w)$, from the first we get that $M,w,\epsilon \models \phi$, and from that and the second $M,w,\epsilon \models R_B \phi$, so that $M,w,B \models \phi$. We therefore easily show that $\phi$ is $*$-valid by induction on the length of resolution sequences.) 

\begin{example} \label{example.rd1rd2}
Model $M$ in Figure \ref{fig2} provides a counterexample to (RD1), and model $M'$ below provides a counterexample to (RD2), where we have again named worlds with valuations of atoms.

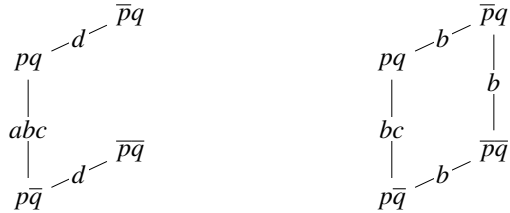
\begin{figure}[h]
\begin{center}
\scalebox{.9}{
\begin{tikzpicture}
%
\node (10) at (0,0) {$p\nq$};
\node (00) at (1.5,.7) {$\np \nq$};
\node (11) at (0,2) {$p q$};
\node (01) at (1.5,2.7) {$\np q$};
\draw (00) -- node[fill=white,inner sep=1pt] {$d$} (10);
\draw (01) -- node[fill=white,inner sep=1pt] {$d$} (11);
\draw (10) -- node[fill=white,inner sep=1pt] {$abc$} (11);
\end{tikzpicture}
\qquad\qquad\qquad\qquad
\begin{tikzpicture}
%
\node (10) at (0,0) {$p\nq$};
\node (00) at (1.5,.7) {$\np \nq$};
\node (11) at (0,2) {$p q$};
\node (01) at (1.5,2.7) {$\np q$};
\draw (00) -- node[fill=white,inner sep=1pt] {$b$} (10);
\draw (01) -- node[fill=white,inner sep=1pt] {$b$} (11);
\draw (00) -- node[fill=white,inner sep=1pt] {$b$} (01);
\draw (10) -- node[fill=white,inner sep=1pt] {$bc$} (11);
\end{tikzpicture}
}
\end{center}
\caption{Epistemic models $M$ (on the left) and $M'$ (on the right) with four agents.}
\end{figure}\label{fig2}

We have that $M,pq,\epsilon \models D_{abc} R_{ab} \lnot K_a q$ but $M,pq,\epsilon \not \models R_{ab}D_{bc}\lnot K_a q$---after resolution $ab$, agents $b$ and $c$ can imagine $a$ has further shared knowledge with $d$, thereby learning that $q$. Therefore (RD1) is invalid.

Looking at $M'$, we have that $M',pq, \epsilon \models K_a R_{bc}K_b p$ but $M',pq,\epsilon \not\models R_{bc}K_a K_b p$ because $a$ is not aware of $b$ and $c$ sharing their knowledge so she does not know $b$ has learnt that $p$. Therefore (RD2) is invalid.
\end{example}

On the other hand, there are now novel validities involving distributed knowledge and resolution.
\begin{lemma}\label{lemmaValiditiesRD} 
 \hfill
 
$(1)\quad$ If $B \cap C \neq \emptyset$, then $\models R_B D_C \phi \imp D_{B \cup C} R_{B.\vec{I}} \phi$ for all $\vec{I}\in \power(A\setminus C)^\ast$.

$(2)\quad$ If $B \cap C \neq \emptyset$, then $\models D_{B\union C} R_{B.\vec{I}}\phi$ for all $\vec{I}\in \power(A\setminus C)^\ast$ implies $\models R_B D_C \phi$.

$(3)\quad$ If $B \cap C = \emptyset$, then $\models R_B D_C \phi \eq D_{C} \phi$. 
\end{lemma}

\medskip We recall once more, comparing to the first and second items jointly, (RD1) $\models R_B D_C \phi \eq D_{B \cup C} R_B \phi$ when $B \cap C \neq \emptyset$, and, comparing to the third item  (RD2) $\models R_B D_C \phi \eq D_C R_B \phi$ when $B \cap C = \emptyset$.
The proofs of the above are fairly straightforward but are omitted, as we will present a complete axiomatization {\bf RAD} in the next section. Then, the third item is (in one direction) an instantiation of the later axiom (RD) for the case that $\vec{G}=B$ and $\vec{H}=\epsilon$, whereas for the first item we have $\vec{G} = B$ as well as $\vec{H}= B.\vec{I}$. The second item instantiates derivation rule (RDI) of {\bf RAD}.

Even for fairly simple (`small') models given a set of agents $A$, and even when everyone's view is $A$ (for all $a$, $\see_a(\vec{G})=A$) there is no bound to uncertainty caused by asynchrony (see the example below). This is different from synchrony, where once everyone's view is $A$, this is commonly known. Two notes: with synchrony, we can also have arbitrary higher-order uncertainty, but at the price of large models (with enchained equivalence classes for different agents); and with asynchrony, we can also have common knowledge that everyone's view is $A$, namely after resolution with $A$ (full synchronization).

\begin{example}
Let $|A| \geq 3$. Consider a resolution sequence $\vec{G}$ without $\emptyset$ and without singleton sets (irrelevant), and without $A$. Suppose towards a contradiction that there is bound to the length $|\vec{G}|$ of $\vec{G}$ after which a further resolution is no longer informative. Let $I \subseteq A$ with $1 < |I| < |A|$. Below we define a (unique) model $M$ and a (unique) formula $\psi \in \lang_{DR}$ that distinguishes $\vec{G}$ from $\vec{G}.I$. Consider the formula
\begin{center}
    $\phi := \Et_{a \in A} K_a p \et \M_a \Et_{b \in A{\setminus}\{a\}} \neg(K_b p \vel K_b \neg p)$
\end{center}
\noindent and a model $M$ consisting of $2n+1$ states namely $\{s^1\} \union \{ s^2_a \mid a \in A\} \union \{ t^2_a \mid a \in A\}$, where the relations $\sim_a$ are the reflexive closure of the conditions: for all $a,b\in A$ with $a \neq b$, $s^2_a \sim_b t^2_a$, and for all $a\in A$, $s^2_a \sim_a s^1$, and where valuation $V(p) = \{s^1\} \union \{ s^2_a \mid a \in A\}$.
We now have that $M, s^1,\epsilon \models \phi$ (namely, already in standard epistemic logic, $M,s^1 \models \phi$). Let $\Sigma$ be the sequence of length $|\vec{G}|$ of agents $a \in A$ such that the last is a member of $I$ but not of $B$ preceding $I$, the before last is a member of $B$ but not of $B'$ preceding $B$, and so on. Let the first $B$ in $\vec{G}.I$ apart from the selected member $a$ in $\Sigma$ also contain an agent $b \neq a$. Now consider \[ \psi := K_\Sigma (K_a K_b p \et K_b K_a p) \] where $K_\Sigma$ abbreviates the stack $K_x\dots K_y$ listing all the members of $\Sigma$. Then $M,s^1,\vec{G}\not\models \psi$ whereas $M,s^1,\vec{G}.I \models \psi$.

For four agents $a,b,c,d$, (so-called `windmill') model $M$ is depicted in Figure \ref{windmillModel}. In such a model we have that, for example, $M,s^1,ab.bc \models K_c (K_a K_b p \et K_b K_a p)$ but $M,s^1,ab \not\models K_c (K_a K_b p \et K_b K_a p)$. For an example where everyone's view is $A$ but uncertainty still remains, consider the sequence of resolutions $abc.bcd.abd$ after which all four agents have accessibility relation $\sim_A$. We however have $M,s^1,abc.bcd.abd \not \models K_cK_aK_dp$, but now $M,s^1,abc.bcd.abd.abc \models K_cK_aK_dp $.
\end{example}

\begin{figure}[h]
\begin{center}
\scalebox{.75}{
\begin{tikzpicture}
\node (1) at (0,0) {$p (s^1)$};
\node (2a) at (1.8,0) {$p$};
\draw[-] (1) -- node[above] {$a$} (2a);
\node (2ah) at (2.5,1.5) {$\neg p$};
\draw[-] (2a) -- node[right] {$bcd$} (2ah);
\node (2b) at (0,1.8) {$p$};
\draw[-] (1) -- node[left] {$b$} (2b);
\node (2bh) at (-1.5,2.5) {$\neg p$};
\draw[-] (2b) -- node[below] {$acd$} (2bh);
\node (2c) at (-1.8,0) {$p$};
\draw[-] (1) -- node[below] {$c$} (2c);
\node (2ch) at (-2.5,-1.5) {$\neg p$};
\draw[-] (2c) -- node[left] {$abd$} (2ch);
\node (2d) at (0,-1.8) {$p$};
\draw[-] (1) -- node[right] {$d$} (2d);
\node (2dh) at (1.5,-2.5) {$\neg p$};
\draw[-] (2d) -- node[below] {$abc$} (2dh);
\end{tikzpicture}
}
\end{center}
\caption{`Windmill' model $M$ with four agents.}
\end{figure}
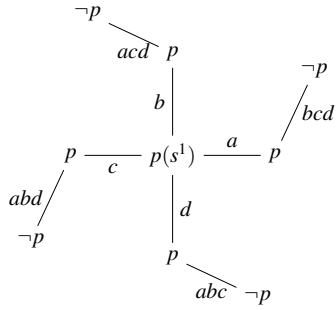\label{windmillModel}

\section{Axiomatization}  \label{sec.axiomatization}

In this section we show soundness and completeness of the axiomatization \textbf{RAD}, that is composed of the axioms and rules given in Table \ref{axiomatisationRAD}. That $\phi \in \lang_{DR}$ is derivable in \textbf{RAD} is denoted $\vdash \phi$. Axiomatization \textbf{RAD} contains an infinitary derivation rule (RDI) using admissible forms, defined below.  As we have an infinitary derivation rule, in the completeness part of the proof we proceed by maximal consistent theories, also defined below, instead of the usual maximal consistent sets. Showing completeness involves `unravelling' a canonical model, in order to get it into the right shape. Before we proceed, let us comment on the differences between axiomatizations {\bf RD} and {\bf RAD}. First, as $R_\epsilon\phi$ is $\phi$ by definition, the {\bf RD} axioms (K$_D$) \dots (DG) involving distributed knowledge are derivable from (in fact, instantiations of) the {\bf RAD} axioms (RK$_D$) \dots (RDG). Second, as resolutions $R_{\vec{G}}$ for a singleton sequence $\vec{G}$ are simply $R_B$ for some $B \subseteq A$, we can also equate the {\bf RD} axioms (RA), (RN) and (RC) with the similarly named {\bf RAD} axioms (where the {\bf RAD} axiom R$\top$ is derivable in {\bf RD} using (NecR)). So the only difference is that the {\bf RD} axioms (RD1) and (RD2) are missing (and they are not theorems, because we have shown through counterexample that they are invalid), instead of which we now have axiom (RD) and rule (RDI) (that allow to derive the validities listed in Lemma \ref{lemmaValiditiesRD}). Finally, {\bf RD} has (NecR) but not {\bf RAD}, where an open question is whether (NecR) is derivable in {\bf RAD}. Let us now proceed.

    \begin{table}
    \centering
    \begin{tabular}{lll}
         (taut)  & all instantiations of propositional tautologies & \\
         (RK$_D$) & $R_{\vec{G}}D_B(\phi \imp \psi)\imp R_{\vec{G}}(D_B\phi \imp D_B \psi)$ & \\
         (RT$_D$) & $R_{\vec{G}}D_B \phi \imp R_{\vec{G}}\phi$ & \\
         (R5$_D$) & $R_{\vec{G}}\lnot D_B\phi \imp R_{\vec{G}}D_B \lnot D_B \phi$ & \\
         (RDG)    & $R_{\vec{G}}D_B \phi \imp R_{\vec{G}}D_C \phi$ & if $B \subseteq C$ \\
         (RA)      & $R_{\vec{G}} p \eq p$ & \\
         (R$\top$) & $R_{\vec{G}}\top$ &  \\
         (RN)      & $R_{\vec{G}} \lnot \phi \eq \lnot R_{\vec{G}} \phi$ & \\
         (RC)      & $R_{\vec{G}} (\phi \land \psi) \eq (R_{\vec{G}} \phi \land R_{\vec{G}} \psi)$ & \\
         (RD)      & $R_{\vec{G}} D_B \phi \imp D_{see_B(\vec{G})}R_{\vec{H}}\phi$ & for all $\vec{H}\approx_B \vec{G}$ \\
         (NecD)    & From $\phi$, infer $D_B \phi$ & \\
         (MP)      & From $\phi$ and $\phi \imp \psi$, infer $\psi$ & \\
         (RDI)     & \multicolumn{2}{l}{From $\alpha(D_{see_B(\vec{G})}R_{\vec{H}}\phi)$ for all $\vec{H}\approx_B \vec{G}$, infer $\alpha(R_{\vec{G}} D_B\phi)$}
    \end{tabular}
    \caption{Axiomatization \textbf{RAD} for the logic of resolving asynchronous distributed knowledge}
    \label{axiomatisationRAD}
\end{table}

    An \emph{admissible form} $\alpha$ is defined by $\alpha ::= \sharp \ |\ (\phi \imp \alpha) \ |\ D_B \alpha$    where $\phi \in \lang_{DR}$ and $B \subseteq A$. The set of admissible forms is denoted $AForm$. An admissible form contains a unique occurrence of $\sharp$. For $\alpha \in AForm$ and $\phi \in \lang_{DR}$, $\alpha(\phi)$ is the formula obtained by replacing $\sharp$ in $\alpha$ by $\phi$.

\begin{lemma}
    If $\vdash \phi \imp \psi$, then $\vdash \alpha(\phi) \imp \alpha(\psi)$.
\end{lemma}

\begin{proof}
    The proof proceeds by straightforward induction on $\alpha$.
\end{proof}

\begin{theorem}[Soundness]
    If $\vdash\phi$  then $\models \phi$.
\end{theorem}

\begin{proof}
    It needs to be shown that all axioms are valid and all rules preserve validity. As for the axioms, their validity is obvious, except for (RD). Let us then show $\models R_{\vec{G}}D_B \phi \imp D_{see_B(\vec{G})}R_{\vec{H}}\phi$ for all $\vec{H}\in \power(A)^\ast$ such that $\vec{G}\approx_B \vec{H}$.
    Suppose there are $\vec{G},\vec{H},B$ and $M= (W,\sim,V)$, $w\in W$ such that $M,w,\epsilon \models R_{\vec{G}}D_B \phi$ and $M,w,\epsilon \not \models D_{see_B(\vec{G})}R_{\vec{H}}\phi$, where $\vec{G}\approx_B \vec{H}$. Then, there are $v\in W$ and $\vec{I}\in \power(A)^\ast$ such that $w \sim_{see_B(\vec{G})} v$ \emph{i.e.} $w \sim^{\vec{G}}_Bv$, $\epsilon \approx_{see_B(\vec{G})} \vec{I}$ and $M,v,\vec{I}\not \models R_{\vec{H}}\phi$, so $v,\vec{I}.\vec{H}\not \models \phi$. Since $\epsilon \approx_{see_B(\vec{G})}\vec{I}$ and $\vec{G}\approx_B \vec{H}$, by Lemma \ref{lemmaCompositionSequences}, $\epsilon.\vec{G}=\vec{G}\approx_B \vec{I}.\vec{H}$. Hence from $M,v,\vec{I}.\vec{H}\not \models \phi$ we can conclude $M,w,\vec{G}\not\models D_B \phi$ \emph{i.e.} $M,w,\epsilon\not\models R_{\vec{G}}D_B \phi$. This contradicts the hypothesis.

    We now turn to the rules. That (MP) and (Nec) preserve validity can be standardly shown. Let us consider (RDI).
    For clarity, we only consider the case where $\alpha = \sharp$, whilst others can be treated similarly.
    Suppose $\models D_{see_B(\vec{G})}R_{\vec{H}}\phi$ for all $\vec{H}\in \power(A)^\ast$ such that $\vec{G} \approx_B \vec{H}$. Suppose also, towards a contradiction, $\not\models R_{\vec{G}}D_B \phi$. Then there are $M=(W,\sim,V), w \in W$ such that $M,w,\epsilon \not\models R_{\vec{G}}D_B \phi$, \emph{i.e.} $M,w,\vec{G}\not \models D_B \phi$. Hence, there are $v\in W, \vec{H} \in \power(A)^\ast$ such that $w\sim_B^{\vec{G}}v$, $\vec{G} \approx_B \vec{H}$ and $M,v,\vec{H}\not \models \phi$, so $M,v,\epsilon \not \models R_{\vec{H}}\phi$. But now, since $w\sim_B^{\vec{G}}v$, also $w\sim_{see_B(\vec{G})}v$. Moreover, $\epsilon \approx_{see_B(\vec{G})} \epsilon$ by definition. Hence, from $M,v,\epsilon \not \models R_{\vec{H}}\phi$ we get $M,w,\epsilon\not \models D_{see_B(\vec{G})} R_{\vec{H}} \phi$, where $\vec{G}\approx_B \vec{H}$. This contradicts the hypothesis. Therefore, $\models R_{\vec{G}}D_B \phi$.
\end{proof}

{\bf Standard frames and semi-standard frames.} \label{section.standard}
Instead of models wherein $\sim_B$ is equal to $\inter_{b \in B} \sim_b$ (and even by definition) we need models wherein $\sim_B$ may be a proper subset of $\inter_{b \in B} \sim_b$. The first we name {\em standard models}, based on {\em standard frames}, whereas the second are {\em semi-standard models} based {\em semi-standard frames}. More precisely, a semi-standard frame is a structure $(W,\sim)$ where $W$ is a non-empty set and, for all $B\subseteq A$, $\sim_B$ is an equivalence relation on $W$ such that for all $B,C\subseteq A$, if $B\subseteq C$ then $\sim_C \subseteq \sim_B$. The completeness proofs involving distributed knowledge often involve `unravelling' a canonical model based on a semi-standard frame into one that is based on a standard frame with the same information content \cite{FaginHV92,AgotnesW17}. In our completeness proof we employ the results relating standard and semi-standard frames obtained in \cite[Section 6]{BalbianiD24} and in particular \cite[Proposition 6.10]{BalbianiD24} that the validities on standard frames and semi-standard frames correspond. Here, we merely adapt this to our setting, where the result to use is that from any model $M = (W,\sim,V)$ based on a semi-standard frame we can construct a model $M' = (W',\sim',V')$ based on a standard frame; in $M'$ the worlds consist of pairs $(w,f)$ where $w \in W$ and $f$ is a function of the set $\mathcal F$ of such functions of type $\power(A) \times A \imp \power(W)$. One can then show that (i) if $(w,f)\sim'^{\vec{G}}_B (v,g)$, then $w\sim^{\vec{G}}_B v$, and that (ii) if $w\sim^{\vec{G}}_B v$, then there is $g\in \mathcal{F}$ such that $(w,f)\sim'^{\vec{G}}_B (v,g)$, for all $\vec{G}\in \power(A)^\ast$. 

Providing a procedure to transform each model based on a semi-standard frame into another modally equivalent model that is based on a standard frame requires the following notions.
For all $B \subseteq A$ and for all $w \in W$, $\lbrack w\rbrack_{B}$ is the equivalence class of $w$ modulo $\sim_{B}$.
For all $X,Y \in \power(W)$, let $X+Y=(X\setminus Y)\cup(Y\setminus X)$\footnote{Note that $(\power(W),\emptyset,W,+,\cap)$ is a Boolean ring.}.
We recall that ${\mathcal F}$ is the set of all functions of type $f:\power(A)\times A\longrightarrow\power(W)$.

Let now $M=(W,\sim,V)$ be a model based on a semi-standard frame. We define the model $M'=(W',\sim',V')$ by $W':=W\times{\mathcal F}$, $(w,f) \in V'(p)$ iff $w\in V(p)$ and for all $B \subseteq A$, $\sim'_{B}$ is the binary relation on $W'$ such that for all $(w,f),(v,g) \in W'$, $(w,f){\sim'_{B}}(v,g)$ iff, for all $C \subseteq A$, the two following conditions hold:

        $- \quad \lbrack w\rbrack_{C}+\Sigma_{a \in C}f(C,a)=\lbrack v\rbrack_{C}+\Sigma_{a \in C}g(C,a)$ 
        
        $- \quad \text{if } a\in B \cap C, \text{ then } f(C,a)=g(C,a) \text{ for all } a \in A$. \\

Obviously, all $\sim'_B$ are equivalence relations, so $(W',\sim')$ is a frame. We now show that it is a standard frame.

\begin{lemma}\label{lemmaTheFrameIsSemiStandard}
    The frame $(W',{\sim'})$ is semi-standard.
\end{lemma}

\begin{proof}
    Let $B,B' \subseteq A$.
    Suppose $B \subseteq B'$.
    Suppose $\sim'_{B'}{\not\subseteq}\sim'_{B}$.
    Hence, there exist $(w,f),(v,g) \in W'$ such that $(w,f){\sim'_{B'}}(w,g)$ and $(w,f){\not\sim'_{B}}(v,g)$.
    Thus, either there exists $C \subseteq A$ such that $\lbrack w\rbrack_{C}+\Sigma_{a \in C}f(C,a) \neq \lbrack v\rbrack_{C}+\Sigma_{a \in C}g(C,a)$, or there exist $C \subseteq A$ and $a \in B \cap C$ such that $f(C,a) \neq g(C,a)$.
    In the former case, since $(w,f){\sim'_{B'}}(v,g)$, therefore $\lbrack w\rbrack_{C}+\Sigma_{a \in C}f(C,a)=\lbrack v\rbrack_{C}+\Sigma_{a \in C}g(C,a)$: a contradiction.
    In the latter case, since $B \subseteq B'$, also $a \in B'\cap C$.
    Since $(w,f){\sim'_{B'}}(v,g)$, therefore $f(C,a)=g(C,a)$: a contradiction.
    Therefore $\sim'_{B'}{\subseteq}\sim'_{B}$.
\end{proof}

\begin{lemma}\label{lemmaTheFrameIsStandard}
    The semi-standard frame $(W',{\sim'})$ is standard.
\end{lemma}

\begin{proof}
    Let $B,B' \subseteq A$.
    Suppose $\sim'_{B\cup B'}{\not\supseteq}\sim'_{B}\cap\sim'_{B'}$.
    Hence, there exist $(w,f),(v,g) \in W'$ such that $(w,f){\not\sim'_{B\cup B'}}(v,g)$, $(w,f){\sim'_{B}}(v,g)$ and $(w,f){\sim'_{B'}}(v,g)$.
    Thus, for all $C \subseteq A$, $\lbrack w\rbrack_{C}+\Sigma_{b \in C}f(C,b)=\lbrack v\rbrack_{C}+\Sigma_{b \in C}g(C,b)$.
    Moreover, for all $C \subseteq A$ and for all $b \in B \cap C$, $f(C,b)=g(C,b)$ and for all $C \subseteq A$ and for all $b \in B' \cap C$, $f(C,b)=g(C,b)$. Since $(w,f){\not\sim'_{B\cup B'}}(v,g)$, there exist $E \subseteq A$ and $a \in (B\cup B')\cap E$ such that $f(E,a) \neq g(E,a)$. But, since $a\in B \cup B'$, either $a \in B$, or $a \in B'$.
    In the former case, since for all $C \subseteq A$ and for all $b \in B \cap C$, $f(C,b)=g(C,b)$, therefore $f(E,a)=g(E,a)$: a contradiction. In the latter case, since for all $C \subseteq A$ and for all $b \in B'\cap C$, $f(C,b)=g(C,b)$, therefore $f(E,a)=g(E,a)$: a contradiction.
    Therefore $\sim'_{B\cup B'}\supseteq\sim'_{B}\cap\sim'_{B'}$.
\end{proof}

\begin{lemma}\label{lemmaSimStandardToSemiStandard}
    For all $w,v\in W, f,g \in \mathcal{F}$ and $B \subseteq A$, if $(w,f)\sim'_B (v,g)$, then $w\sim_B v$.
\end{lemma}

\begin{proof}
    Let $B \subseteq A$ and $(w,f),(v,g) \in W'$ be such that $(w,f) \sim'_{B}(v,g)$.
    Hence, for all $C \subseteq A$, $\lbrack w\rbrack_{C}+\Sigma_{a \in C}f(C,a)=\lbrack v\rbrack_{C}+\Sigma_{a \in C}g(C,a)$.
    Moreover, for all $a \in B \cap C$, $f(C,a)=g(C,a)$.
    Thus, for $C = B$ we get $\lbrack w\rbrack_{B}+\Sigma_{a \in B}f(B,a)=\lbrack v\rbrack_{B}+\Sigma_{a \in B}g(B,a)$.
    Moreover, for all $a \in B$, $f(B,a)=g(B,a)$.
    Consequently, $\lbrack w\rbrack_{B}=\lbrack v\rbrack_{B}$.
    Hence, $w \sim_{B} v$.
\end{proof}

\begin{corollary}\label{coSimStandardToSemiStandard}
     For all $w,v\in W, f,g \in \mathcal{F}$ and $\vec{G}\in \power(A)^\ast, B \subseteq A$, if $(w,f)\sim'^{\vec{G}}_B (v,g)$, then $w\sim^{\vec{G}}_B v$.
\end{corollary}

\begin{lemma}\label{lemmaSimSemiStandardToStandard}
    For all $w,v\in W, f \in \mathcal{F}$ and $B \subseteq A$, if $w\sim_B v$, then there is $g\in \mathcal{F}$ such that $(w,f)\sim'_B (v,g)$.
\end{lemma}

\begin{proof}
    Let $w,v\in W$. Suppose $w\sim_B v$, \emph{i.e.} $[w]_B=[v]_B$. Since $(W',\sim')$ is semi-standard, for all $C \subseteq A$, if $C \subseteq B$, then $w \sim_C v$ and $[w]_C=[v]_C$.
    To construct $g:\power(A)\times A\longrightarrow\power(W)$, consider an enumeration $\lbrace a_1, a_2, \dots, a_k \rbrace$ of the agents in $A$. For all $C \subseteq A$ and for all $a \in A$, we define $g(C,a)$ as follows:
    
    $-$ if $a \in B\cap C$ then $g(C,a):=f(C,a)$
    
    $-$ if $a \in C\setminus B$, let $1 \leq i \leq k$ be such that $a = a_i$. Now, if $i = min\lbrace 1 \leq j \leq k \ |\ a_j \in C\setminus B \rbrace$, then $g(C,a):= [w]_C+[v]_C + \sum_{b\in C\setminus B}f(C,b)$; otherwise $g(C,a):=\emptyset$.
    
   $-$ if $a\notin C$ then $g(C,a):=\emptyset$.

\noindent The reader may easily verify that $(w,f){\sim'_{B}}(v,g)$.
\end{proof}

\begin{corollary}\label{coSimSemiStandardToStandard}
    For all $w,v\in W, f \in \mathcal{F}$ and $\vec{G}\in \power(A)^\ast, B \subseteq A$, if $w\sim^{\vec{G}}_B v$, then there is $g\in \mathcal{F}$ such that $(w,f)\sim'^{\vec{G}}_B (v,g)$.
\end{corollary}

We can now prove the following lemma that allows one to transform any semi-standard model into a standard model while preserving the satisfaction relation.

\begin{lemma}\label{lemmaSemiStandardToStandard}
     Let semi-standard $M=(W,\sim,V)$ be given and standard $M'=(W',\sim',V')$ be constructed from $M$ as above. For all $\phi \in \lang_{DR}$, $w\in W$, $f\in \mathcal{F}$ and $\vec{G}\in \power(A)^\ast$: $M,w,\vec{G} \models \phi$ iff $M',(w,f),\vec{G} \models \phi$.
\end{lemma}

\begin{proof}
    The proof proceeds by induction on $\phi$. Let $w\in W, f \in \mathcal{F}$ and $\vec{G}\in \power(A)^\ast$.
    \begin{itemize}
        \item If $\phi = p$, then obviously $M,w,\vec{G}\models p \Eq M',(w,f),\vec{G} \models p $.
        \item If $\phi = \top$, then obviously $M,w,\vec{G}\models \top$ and $M',(w,f),\vec{G} \models \top $.
        \item If $\phi = \lnot \psi$, $\phi = \psi \land \psi$ or $\phi = R_B \psi$, we conclude by applying the inductive hypothesis. This is straightforward.
        \item If $\phi = D_B \psi$, suppose $M,w,\vec{G}\models D_B \psi$ and $M',(w,f),\vec{G}\not \models D_B \psi$. Then, there are $(v,g) \in W'$ and $\vec{H} \in \power(A)^\ast$ such that $(w,f) \sim'^{\vec{G}}_B(v,g)$, $\vec{G}\approx_B \vec{H}$ and $M',(v,g),\vec{H}\not\models \psi$. By induction hypothesis, then $M,v,\vec{H}\not\models \psi$. Since $(w,f) \sim'^{\vec{G}}_B(v,g)$, by Corollary \ref{coSimStandardToSemiStandard}, $w\sim_B^{\vec{G}}v$. Moreover, $\vec{G}\approx_B \vec{H}$. Therefore $M,w,\vec{G}\not\models D_B \phi$. This contradicts the hypothesis.
        
        Suppose now $M,w,\vec{G}\not \models D_B \phi$ and $M',(w,f),\vec{G} \models D_B \psi$. Then, there are $v\in W, \vec{H} \in \power(A)^\ast$ such that $w\sim_B^{\vec{G}}v$, $\vec{G}\approx_B \vec{H}$ and $M,v,\vec{H} \not \models \psi$. Now, by Corollary \ref{coSimSemiStandardToStandard}, there is $g\in \mathcal{F}$ such that $(w,f)\sim'^{\vec{G}}_B(v,g)$. By induction hypothesis, from $M,v,\vec{H}\not\models \psi$ we get $M',(v,g),\vec{H}\not\models \psi$, where $(w,f)\sim_B'^{\vec{G}}(v,g)$ and $\vec{G}\approx_B \vec{H}$. Hence $M',(w,f),\vec{G}\not\models D_B \psi$: a contradiction.
    \end{itemize}
\end{proof}

\begin{proposition}\label{propStandardToSemiStandard}
    For all $\phi \in \lang_{DR}$, if $\phi$ is valid on standard frames then $\phi$ is valid on semi-standard frames.
\end{proposition}

\begin{proof}
    By Lemma~\ref{lemmaSemiStandardToStandard}.
\end{proof}

We now move to the proof of the completeness of \textbf{RAD}. This requires further terminology.

    A \emph{theory} $T$ is a set of formulas $\phi \in \lang_{DR}$ such that $T$ contains all formulas derivable in \textbf{RAD}; $T$ is closed under (MP) and $T$ is closed under (RDI). A theory $T$ is \emph{consistent} if it does not contain $\bot$. Furthermore, $T$ is \emph{maximal consistent} if $T$ is consistent and, for all theories $T'$, if $T \subsetneq T'$, then $T'$ is not consistent.
Note that the only inconsistent theory is the theory containing all formulas. 
For $T$ a theory, $\chi \in \lang_{DR}$ and $B\subseteq A$ we define $T + \chi:= \lbrace \phi \ |\ \chi \imp \phi \in T \rbrace$ and $D_BT := \lbrace \phi \ | \ D_B \phi \in T  \rbrace$.

\begin{lemma}\label{lemmaTheories}
    Let $T$ be a theory, $\chi \in \lang_{DR}$ and $B\subseteq A$. Then, $T + \chi$ and $D_BT$ are also theories. Moreover, $T\subseteq T + \chi$ and $\chi \in T+\chi$. Finally, if $\lnot\chi \notin T$, $T+\chi$ is consistent, and if $\chi \notin T$, $T+ \lnot\chi$ is consistent.
\end{lemma}

\begin{proof}
    This is standard.
\end{proof}

\begin{lemma}\label{lemmaRDI}
    Let $\vec{G} \in \power(A)^\ast, B\subseteq A$ and $\phi \in \lang_{DR}$. If $T$ is a theory and $\alpha(R_{\vec{G}}D_B \phi) \notin T$, then there is $\vec{H} \in \power(A)$ such that $\vec{G}\approx_B \vec{H}$ and $\alpha(D_{see_B(\vec{G})}R_{\vec{H}}\phi)\notin T$.
\end{lemma}

\begin{proof}
    This is straightforward, since $T$ is closed under (RDI).
\end{proof}

\begin{lemma}[Lindenbaum's Lemma]\label{LindenbaumLemma}
    If $T$ is a consistent theory, then there is a maximal consistent theory $\Sigma$ such that $T \subseteq \Sigma$.
\end{lemma}

\begin{proof}
    Let $T$ be a consistent theory and $\lbrace \phi_0,\phi_1, \cdots \rbrace$ an enumeration of formulas in $\lang_{DR}$. For all $k\in \mathbb{N}$, we define the theory $T_k$ as follows (from the construction and Lemma \ref{lemmaTheories} it follows that these $T_k$ are theories, because $T$ is a theory):
    \begin{align*}
        T_0      &:= T \\
        T_{k+1}  &:= \begin{cases}
            T_k             &\text{if } \lnot \phi_k \in T_k \\
            T_k + \phi_k    &\text{if } \lnot \phi_k \notin T_k \text{ and } \phi_k \text{ does not have the shape } \lnot \alpha(R_{\vec{G}}D_B \psi) \\
            T_k + \lnot\alpha(D_{see_B(\vec{G})}R_{\vec{H}}\psi) &\text{if } \lnot \phi_k \notin T_k \text{ and } \phi_k \text{ does have the shape } \lnot \alpha(R_{\vec{G}}D_B \psi), \\ & \text{for some } \vec{H}\approx_B \vec{G} \text{ such that } \alpha(D_{see_B(\vec{G})}R_{\vec{H}}\psi)\notin T_k
        \end{cases}
    \end{align*}
    Note that in the third case of the inductive step, such a sequence $\vec{H}$ is known to exist by Lemma \ref{lemmaRDI}. 

    Let now $\Sigma:= \bigcup_{k \geq 0} T_k$. It needs to be shown that $\Sigma$ is a maximal consistent theory. First note that, by construction, for all $\phi \in \lang_{DR}$, either $\phi \in \Sigma$ or $\lnot \phi \in \Sigma$.

    We first show that $\Sigma$ is a theory. That $\Sigma$ contains \textbf{RAD} and is closed by modus ponens is straightforward. Suppose now there are $\vec{G}\in \power(A)^\ast, B \in \power(A), \phi \in \lang_{DR}$ and $\alpha \in AForm$ such that $\alpha(D_{see_B(\vec{G})}R_{\vec{H}}\phi) \in \Sigma$ for all $\vec{H}\approx_B \vec{G}$. Suppose also $R_{\vec{G}}D_B \phi \notin \Sigma$. Then $\lnot R_{\vec{G}}D_B\phi \in \Sigma$. Let $k \in \mathbb{N}$ such that $\phi_k = \lnot R_{\vec{G}}D_B\phi$. Then, since $\phi_k \in \Sigma$, $\lnot \phi_k \notin \Sigma$ so in particular $\lnot\phi_k \notin T_k$. Since $T_k$ is a theory and $R_{\vec{G}}D_B \phi \notin T_k$, by Lemma \ref{lemmaRDI}, there is $\vec{H}\approx_B \vec{G}$ such that $\alpha(D_{see_B(\vec{G})}R_{\vec{H}}\phi)\notin T_k$. Hence, by construction, $T_{k+1} = T_k + \lnot \alpha(D_{see_B(\vec{G})}R_{\vec{H}}\phi)$. Therefore $\alpha(D_{see_B(\vec{G})}R_{\vec{H}}\phi) \notin \Sigma$, which contradicts the hypothesis.

    Furthermore, $\Sigma$ is consistent because so is each $T_k$. We now show that $\Sigma$ is maximal consistent. Let $\Delta$ be a theory strictly containing $\Sigma$. Then, there is $\phi_k \in \lang_{DR}$ such that $\phi_k \notin \Sigma$ and $\phi_k \in \Delta$. Since $\phi_k \notin \Sigma$, $\lnot \phi_k \in \Sigma \subset \Delta$. Hence, $\lnot \phi_k \in \Delta$, so $\Delta$ is not consistent.
\end{proof}

We now (re)introduce the relation $\sim_B$. Let $\Sigma,\Delta$ be two maximal consistent theories, $B\subseteq A$. Then $\Sigma \sim_B \Delta$ iff $D_B \Sigma \subseteq \Delta$. 
Note that by definition $\sim_{B \cup C} \ \subseteq \ \sim_B \cap \sim_C$, but equality is not guaranteed.

\begin{lemma}[Existence Lemma]\label{ExistenceLemma}
    Let $\Sigma$ be a maximal consistent theory, $B\subseteq A,\phi \in \lang_{DR}$. If $\hat{D}_B \phi \in \Sigma$, then there is a maximal consistent theory $\Delta$ such that $\Sigma \sim_B \Delta$ and $\phi \in \Delta$.
\end{lemma}

\begin{proof}
    Suppose $\hat{D}_B \phi \in \Sigma$. Then $\lnot D_B \lnot \phi \in \Sigma$ so $D_B\lnot \phi \notin \Sigma$. Hence $\lnot\phi \notin D_B\Sigma$.  So, by Lemma \ref{lemmaTheories}, $D_B\Sigma + \phi$ is consistent. Now, by Lindenbaum's Lemma, we can extend $D_B\Sigma +\phi$ to a maximal consistent theory $\Delta$ such that $D_B\Sigma + \phi \subseteq \Delta$. By Lemma \ref{lemmaTheories} again, $D_B\Sigma \subseteq D_B\Sigma + \phi$ so $D_B\Sigma \subseteq \Delta$. Hence $\Sigma \sim_B \Delta$. Moreover,  $\phi \in D_B\Sigma + \phi$, so $\phi \in \Delta$.
\end{proof}


\begin{definition}[Canonical model]
    The \emph{canonical model} $M=(W,\lbrace\sim_B\rbrace_{B\subseteq A},V)$ is defined by
    
        $-$ $W= \lbrace \Sigma \ | \ \Sigma \text{ is a maximal consistent theory} \rbrace$;
        
        $-$ $\Sigma \sim_B \Delta$ if, and only if, $D_B \Sigma \subseteq \Delta$;
        
        $-$ $V(p) =\lbrace \Sigma \in W \ | \ p \in \Sigma \rbrace $.

\end{definition}

Notice the canonical model is defined on a semi-standard frame.

\begin{lemma}[Truth Lemma]\label{TruthLemma}
    For all $\phi \in \lang_{DR}$, $\vec{G}\in \power(A)^\ast$ and $\Sigma \in W$: $R_{\vec{G}} \phi \in \Sigma$ iff $\Sigma, \vec{G}\models \phi$.
\end{lemma}

\begin{proof}
    The proof proceeds by induction on $\phi$.
    \begin{itemize}
        \item If $\phi = p$: by (RA) we get: $R_{\vec{G}}p \in \Sigma \Eq p \in \Sigma \Eq \Sigma,\vec{G}\models p$.

        \item If $\phi = \top$: $\Sigma,\vec{G}\models\top$ and, by (R$\top$), $R_{\vec{G}}\top \in \Sigma$.
        
        \item If $\phi = \lnot \psi$: by (RN) we get: $R_{\vec{G}}\lnot \psi \in \Sigma \Eq \lnot R_{\vec{G}}\psi \in \Sigma \Eq R_{\vec{G}}\psi \notin \Sigma \overset{(IH)}{\Eq} \Sigma,\vec{G}\not \models \psi \Eq \Sigma,\vec{G}\models \lnot \psi$.

        \item If $\phi = \psi \land \chi$: by (RC) we get $R_{\vec{G}} (\psi \land \chi) \in \Sigma \Eq R_{\vec{G}}\psi \land R_{\vec{G}}\chi \in \Sigma \Eq R_{\vec{G}}\psi \in \Sigma \text{ and } R_{\vec{G}}\chi \in \Sigma \overset{(IH)}{\Eq} \Sigma,\vec{G} \models \psi \text{ and } \Sigma,\vec{G} \models \chi  \Eq \Sigma,\vec{G}\models \psi \land \chi$.

        \item If $\phi = R_B \psi$: $R_{\vec{G}}R_B \psi \in \Sigma \Eq R_{\vec{G}.B} \psi \in \Sigma \overset{(IH)}{\Eq}\Sigma,\vec{G}.B \models \psi \Eq \Sigma,\vec{G}\models R_B \psi$.

        \item If $\phi = D_B \psi$: suppose first $R_{\vec{G}}D_B \psi \in \Sigma$ and $\Sigma,\vec{G}\not \models D_B \psi$. Then there are $\Delta,\vec{H}$ such that $\Sigma \sim_B^{\vec{G}}\Delta$, $\vec{G}\approx_B \vec{H}$ and $\Delta,\vec{H}\not\models \psi$. By induction hypothesis, then, $R_{\vec{H}}\psi \notin \Delta$. Now, since $R_{\vec{G}}D_B \phi \in \Sigma$ and, by (RD) $R_{\vec{G}}D_B \psi \imp D_{see_B(\vec{G})}R_{\vec{H}}\psi \in \Sigma$ (because $\vec{G}\approx_B \vec{H}$), so $D_{see_B(\vec{G})}R_{\vec{H}}\psi \in \Sigma$.  Moreover, since $\Sigma \sim_B^{\vec{G}} \Delta$, also $\Sigma \sim_{see_B(\vec{G})}\Delta$, so $D_{see_B(\vec{G})}\Sigma \subseteq \Delta$. Then $R_{\vec{H}}\psi \in \Delta$, which contradicts $R_{\vec{H}}\psi \notin \Delta$.

        Suppose now $R_{\vec{G}}D_B \psi \notin \Sigma$ and $\Sigma,\vec{G} \models D_B \psi$. Since $R_{\vec{G}}D_B \psi \notin \Sigma$ and $\Sigma$ is closed under (RDI), there is $\vec{H}\in \power(A)^\ast$ such that $\vec{G}\approx_B \vec{H}$ and $D_{see_B(\vec{G})}R_{\vec{H}}\psi \notin \Sigma$. Then, $\lnot D_{see_B(\vec{G})}R_{\vec{H}}\psi \in \Sigma$, so  $ \hat{D}_{see_B(\vec{G})}\lnot R_{\vec{H}}\psi \in \Sigma$. By the Existence Lemma, there is a maximal consistent theory $\Delta$ such that $\Sigma \sim_{see_B(\vec{G})}\Delta$ and $\lnot R_{\vec{H}}\psi \in \Delta$, so $R_{\vec{H}}\psi \notin \Delta$. By induction hypothesis then, $\Delta,\vec{H}\not \models \psi$. But since $\Sigma \sim_{see_B(\vec{G})}\Delta$, also $\Sigma \sim_B^{\vec{G}}\Delta$ and $\vec{G}\approx_B \vec{H}$, so $\Delta,\vec{H}\not \models \psi$ implies $\Sigma,\vec{G}\not \models D_B \psi$, which contradicts the hypothesis.        
    \end{itemize}
    
\end{proof}

\begin{theorem}[Completeness]
    For all $\phi \in \lang_{DR}$, if $\models \phi$, then $\vdash \phi$ .
\end{theorem}

\begin{proof}
    Let $\phi \in \lang_{DR}$. We show the contrapositive: if $\not\vdash \phi$, then $\not \models \phi$. Suppose then $\not\vdash \phi$. Then $\phi \notin $ \textbf{RAD} so \textbf{RAD}$+ \lnot \phi$ is consistent: we extend it to a maximal consistent theory $\Sigma$, by Lindenbaum's Lemma. Therefore, $\lnot \phi = R_\epsilon \lnot \phi \in \Sigma$. By the Truth Lemma, $\Sigma,\epsilon \models \lnot \phi$ so $\Sigma,\epsilon \not \models \phi$. Hence $\phi$ is not valid on semi-standard frames. By Proposition \ref{propStandardToSemiStandard}, then, $\phi$ is not valid on standard frames. Therefore $\not \models \phi$. 
\end{proof}

\section{Conclusion and further research}  \label{sec.further}

We presented a logic of resolving asynchronous distributed knowledge, compared it to the logic of resolving (synchronous) distributed knowledge known from the literature, and provided an infinitary axiomatization for our asynchronous logic. There are a fair number of open questions about this novel logic.
(i) We defined two notions of validity, with respect to the empty resolution sequence and with respect to arbitrary resolution sequences, but we do not know whether these define the same set of validities.
(ii) We gave an infinitary axiomatization, but we have no proof that a finitary axiomatization does not exist.
(iii)  Is necessitation of resolution validity preserving (resp.\ an admissible derivation rule)?
(iv) Is satisfiability decidable, and if so, what is the complexity?
(v) Given the embedding of the synchronous into the asynchronous semantics, what is the expressivity hierarchy comparing the language fragments with synchronous and with asynchronous distributed knowledge, and with or without resolution? It seems fairly straightforward to show that resolving asynchronous distributed knowledge is more expressive than resolving synchronous distributed knowledge, or at least on the level of so-called update expressivity that describes relations between pointed epistemic models instead of properties of pointed epistemic models in the case of formula expressivity. Concerning the latter, clearly, a sequence of two resolutions $ab.bc$ does not correspond to a single resolution $B$ for some $B \subseteq A$, wherein all agents in $B$ are equally well informed. A more interesting question is whether resolving asynchronous distributed knowledge is more expressive than logical semantics of synchronous distributed knowledge with more involved dynamics than resolution, such as \cite{Baltag20,baltagsmets.aiml:2024,cdrv:2023}. For example, a sequence of resolutions $ab.bc$ corresponds to a single communication graph (as in \cite{cdrv:2023}) namely were $b$ and $c$ receive information from $a,b,c$ whereas $a$ only receives information from $a,b$, and, for a more involved example, the uncertainty of an agent $d$ between resolutions $ab$ and $ab.bc$ can be simulated as non-public information exchange as in \cite{baltagsmets.aiml:2024}. We wish to investigate that in the future.
(vi) Can we determine when a resolution in a given sequence is redundant (because uninformative), such as (immediately) repeating the same resolution?
(vii) We wish to extend the language and semantics with common knowledge.

\paragraph*{Acknowledgements.} We thank the referees for their reviews: their useful suggestions have been essential for improving the readability of a preliminary version of this paper.

\bibliographystyle{eptcs}
\bibliography{biblio2026}

\providecommand{\noopsort}[1]{}
\begin{thebibliography}{10}
\providecommand{\bibitemdeclare}[2]{}
\providecommand{\surnamestart}{}
\providecommand{\surnameend}{}
\providecommand{\urlprefix}{Available at }
\providecommand{\url}[1]{\texttt{#1}}
\providecommand{\href}[2]{\texttt{#2}}
\providecommand{\urlalt}[2]{\href{#1}{#2}}
\providecommand{\doi}[1]{doi:\urlalt{https://doi.org/#1}{#1}}
\providecommand{\eprint}[1]{arXiv:\urlalt{https://arxiv.org/abs/#1}{#1}}
\providecommand{\bibinfo}[2]{#2}

\bibitemdeclare{article}{AgotnesW17}
\bibitem{AgotnesW17}
\bibinfo{author}{T.~\surnamestart {\AA}gotnes\surnameend} \&
  \bibinfo{author}{Y.N. \surnamestart W{\'{a}}ng\surnameend}
  (\bibinfo{year}{2017}): \emph{\bibinfo{title}{Resolving distributed
  knowledge}}.
\newblock {\slshape \bibinfo{journal}{Artif. Intell.}} \bibinfo{volume}{252},
  pp. \bibinfo{pages}{1--21}, \doi{10.1016/j.artint.2017.07.002}.

\bibitemdeclare{article}{AlechinaBS12}
\bibitem{AlechinaBS12}
\bibinfo{author}{N.~\surnamestart Alechina\surnameend},
  \bibinfo{author}{P.~\surnamestart Balbiani\surnameend} \&
  \bibinfo{author}{D.~\surnamestart Shkatov\surnameend} (\bibinfo{year}{2012}):
  \emph{\bibinfo{title}{Modal logics for reasoning about infinite unions and
  intersections of binary relations}}.
\newblock {\slshape \bibinfo{journal}{J. Appl. Non Class. Logics}}
  \bibinfo{volume}{22}(\bibinfo{number}{4}), pp. \bibinfo{pages}{275--294},
  \doi{10.1080/11663081.2012.705960}.

\bibitemdeclare{inproceedings}{AptKW17}
\bibitem{AptKW17}
\bibinfo{author}{K.R. \surnamestart Apt\surnameend},
  \bibinfo{author}{E.~\surnamestart Kopczynski\surnameend} \&
  \bibinfo{author}{D.~\surnamestart Wojtczak\surnameend}
  (\bibinfo{year}{2017}): \emph{\bibinfo{title}{On the Computational Complexity
  of Gossip Protocols}}.
\newblock In: {\slshape \bibinfo{booktitle}{Proceedings\ of the 26th {IJCAI}}},
  pp. \bibinfo{pages}{765--771}, \doi{10.24963/ijcai.2017/106}.

\bibitemdeclare{incollection}{Areces:tenCate:2007}
\bibitem{Areces:tenCate:2007}
\bibinfo{author}{C.~\surnamestart Areces\surnameend} \&
  \bibinfo{author}{B.~\surnamestart ten Cate\surnameend}
  (\bibinfo{year}{2007}): \emph{\bibinfo{title}{Hybrid Logics}}.
\newblock In \bibinfo{editor}{J.~\surnamestart van Benthem\surnameend},
  \bibinfo{editor}{P.~\surnamestart Blackburn\surnameend} \&
  \bibinfo{editor}{F.~\surnamestart Wolter\surnameend}, editors: {\slshape
  \bibinfo{booktitle}{The Handbook of Modal Logic}},
  \bibinfo{publisher}{Elsevier}, \bibinfo{address}{Amsterdam, The Netherlands},
  \doi{10.1016/S1570-2464(07)80017-6}.

\bibitemdeclare{inproceedings}{BalbianiD24}
\bibitem{BalbianiD24}
\bibinfo{author}{P.~\surnamestart Balbiani\surnameend} \&
  \bibinfo{author}{H.~\surnamestart van Ditmarsch\surnameend}
  (\bibinfo{year}{2024}): \emph{\bibinfo{title}{Towards Dynamic Distributed
  Knowledge}}.
\newblock In \bibinfo{editor}{A.~\surnamestart Ciabattoni\surnameend},
  \bibinfo{editor}{D.~\surnamestart Gabelaia\surnameend} \&
  \bibinfo{editor}{I.~\surnamestart Sedl{\'{a}}r\surnameend}, editors:
  {\slshape \bibinfo{booktitle}{Proceedings\ of {Advances in Modal Logic}}},
  \bibinfo{publisher}{College Publications}, pp. \bibinfo{pages}{125--146}.

\bibitemdeclare{article}{balbianietal:2003}
\bibitem{balbianietal:2003}
\bibinfo{author}{P.~\surnamestart Balbiani\surnameend} \&
  \bibinfo{author}{D.~\surnamestart Vakarelov\surnameend}
  (\bibinfo{year}{2003}): \emph{\bibinfo{title}{{PDL} with Intersection of
  Programs: A Complete Axiomatization}}.
\newblock {\slshape \bibinfo{journal}{Journal of Applied Non-Classical Logics}}
  \bibinfo{volume}{13(3-4)}, pp. \bibinfo{pages}{231--276},
  \doi{10.3166/jancl.13.231-276}.

\bibitemdeclare{incollection}{baltagetal.hintikka:2018}
\bibitem{baltagetal.hintikka:2018}
\bibinfo{author}{A.~\surnamestart Baltag\surnameend},
  \bibinfo{author}{R.~\surnamestart Boddy\surnameend} \&
  \bibinfo{author}{S.~\surnamestart Smets\surnameend} (\bibinfo{year}{2018}):
  \emph{\bibinfo{title}{Group knowledge in interrogative epistemology}}.
\newblock In \bibinfo{editor}{H.~\surnamestart van Ditmarsch\surnameend} \&
  \bibinfo{editor}{G.~\surnamestart Sandu\surnameend}, editors: {\slshape
  \bibinfo{booktitle}{Jaakko Hintikka on Knowledge and Game-Theoretical
  Semantics}}, \bibinfo{series}{Outstanding contributions to logic 12},
  \bibinfo{publisher}{Springer}, pp. \bibinfo{pages}{131--164},
  \doi{10.1007/978-3-319-62864-6_5}.

\bibitemdeclare{inproceedings}{baltagetal:1998}
\bibitem{baltagetal:1998}
\bibinfo{author}{A.~\surnamestart Baltag\surnameend}, \bibinfo{author}{L.S.
  \surnamestart Moss\surnameend} \& \bibinfo{author}{S.~\surnamestart
  Solecki\surnameend} (\bibinfo{year}{1998}): \emph{\bibinfo{title}{The Logic
  of Public Announcements, Common Knowledge, and Private Suspicions}}.
\newblock In: {\slshape \bibinfo{booktitle}{Proceedings\ of 7th TARK}}, pp.
  \bibinfo{pages}{43--56}, \doi{10.1007/978-3-319-20451-2_38}.

\bibitemdeclare{article}{BaltagS13}
\bibitem{BaltagS13}
\bibinfo{author}{A.~\surnamestart Baltag\surnameend} \&
  \bibinfo{author}{S.~\surnamestart Smets\surnameend} (\bibinfo{year}{2013}):
  \emph{\bibinfo{title}{Protocols for belief merge: Reaching agreement via
  communication}}.
\newblock {\slshape \bibinfo{journal}{Log. J. {IGPL}}}
  \bibinfo{volume}{21}(\bibinfo{number}{3}), pp. \bibinfo{pages}{468--487},
  \doi{10.1093/JIGPAL/JZS049}.

\bibitemdeclare{inproceedings}{Baltag20}
\bibitem{Baltag20}
\bibinfo{author}{A.~\surnamestart Baltag\surnameend} \&
  \bibinfo{author}{S.~\surnamestart Smets\surnameend} (\bibinfo{year}{2020}):
  \emph{\bibinfo{title}{Learning What Others Know}}.
\newblock In: {\slshape \bibinfo{booktitle}{Proceedings\ of 23rd {LPAR}}},
  {\slshape \bibinfo{series}{EPiC Series in Computing}}~\bibinfo{volume}{73},
  pp. \bibinfo{pages}{90--119}, \doi{10.29007/plm4}.

\bibitemdeclare{unpublished}{baltagsmets.aiml:2024}
\bibitem{baltagsmets.aiml:2024}
\bibinfo{author}{A.~\surnamestart Baltag\surnameend} \&
  \bibinfo{author}{S.~\surnamestart Smets\surnameend} (\bibinfo{year}{2024}):
  \emph{\bibinfo{title}{Logics for Data Exchange and Communication}}.
\newblock \bibinfo{note}{Proceedings of the 15th {Advances in Modal Logic}
  Prague}.

\bibitemdeclare{inproceedings}{stefanjohan:2009}
\bibitem{stefanjohan:2009}
\bibinfo{author}{J.~\surnamestart van Benthem\surnameend} \&
  \bibinfo{author}{S.~\surnamestart Minica\surnameend} (\bibinfo{year}{2009}):
  \emph{\bibinfo{title}{Toward a Dynamic Logic of Questions}}.
\newblock In \bibinfo{editor}{X.~\surnamestart He\surnameend},
  \bibinfo{editor}{J.F. \surnamestart Horty\surnameend} \&
  \bibinfo{editor}{E.~\surnamestart Pacuit\surnameend}, editors: {\slshape
  \bibinfo{booktitle}{Logic, Rationality, and Interaction. Proceedings\ of LORI
  2009}}, \bibinfo{series}{LNCS 5834}, \bibinfo{publisher}{Springer}, pp.
  \bibinfo{pages}{27--41}, \doi{10.1007/s10992-012-9233-7}.

\bibitemdeclare{article}{Blackburn:Seligman:1995}
\bibitem{Blackburn:Seligman:1995}
\bibinfo{author}{P.~\surnamestart Blackburn\surnameend} \&
  \bibinfo{author}{J.~\surnamestart Seligman\surnameend}
  (\bibinfo{year}{1995}): \emph{\bibinfo{title}{Hybrid languages}}.
\newblock {\slshape \bibinfo{journal}{Journal of Logic, Language, and
  Information}} \bibinfo{volume}{4}, pp. \bibinfo{pages}{251--272},
  \doi{10.1007/BF01049415}.

\bibitemdeclare{mastersthesis}{boddy:2014}
\bibitem{boddy:2014}
\bibinfo{author}{R.~\surnamestart Boddy\surnameend} (\bibinfo{year}{2014}):
  \emph{\bibinfo{title}{Epistemic Issues and Group Knowledge}}.
\newblock Master's thesis, \bibinfo{school}{ILLC, University of Amsterdam}.
\newblock \urlprefix\url{https://eprints.illc.uva.nl/id/document/2156/}.
\newblock \bibinfo{note}{Master of Logic Series MoL-2014-03}.

\bibitemdeclare{mastersthesis}{carrington:2013}
\bibitem{carrington:2013}
\bibinfo{author}{R.~\surnamestart Carrington\surnameend}
  (\bibinfo{year}{2013}): \emph{\bibinfo{title}{Learning and Knowledge in
  Social Networks}}.
\newblock Master's thesis, \bibinfo{school}{ILLC, University of Amsterdam}.
\newblock \urlprefix\url{https://eprints.illc.uva.nl/id/document/2124/}.
\newblock \bibinfo{note}{Master of Logic Series MoL-2013-18}.

\bibitemdeclare{article}{cdrv:2023}
\bibitem{cdrv:2023}
\bibinfo{author}{A.~\surnamestart Casta{\~{n}}eda\surnameend},
  \bibinfo{author}{H.~\surnamestart van Ditmarsch\surnameend},
  \bibinfo{author}{D.A. \surnamestart Rosenblueth\surnameend} \&
  \bibinfo{author}{D.A. \surnamestart Vel\'azquez\surnameend}
  (\bibinfo{year}{2023}): \emph{\bibinfo{title}{Communication Pattern Logic:
  Epistemic and Topological Views}}.
\newblock {\slshape \bibinfo{journal}{Journal of Philosophical Logic}},
  \doi{10.1007/s10992-023-09713-8}.

\bibitemdeclare{article}{armandoetal.tark:2023}
\bibitem{armandoetal.tark:2023}
\bibinfo{author}{A.~\surnamestart Casta{\~{n}}eda\surnameend},
  \bibinfo{author}{H.~\surnamestart van Ditmarsch\surnameend},
  \bibinfo{author}{D.A. \surnamestart Rosenblueth\surnameend} \&
  \bibinfo{author}{D.A. \surnamestart Vel\'azquez\surnameend}
  (\bibinfo{year}{2023}): \emph{\bibinfo{title}{Comparing the update
  expressivity of communication patterns and action models}}.
\newblock {\slshape \bibinfo{journal}{Electronic Proceedings in Theoretical
  Computer Science}} \bibinfo{volume}{379}, pp. \bibinfo{pages}{157--172},
  \doi{10.4204/EPTCS.379.14}.

\bibitemdeclare{article}{ChristoffGratzlRoy2022}
\bibitem{ChristoffGratzlRoy2022}
\bibinfo{author}{Z.~\surnamestart Christoff\surnameend},
  \bibinfo{author}{N.~\surnamestart Gratzl\surnameend} \&
  \bibinfo{author}{O.~\surnamestart Roy\surnameend} (\bibinfo{year}{2022}):
  \emph{\bibinfo{title}{Priority Merge and Intersection Modalities}}.
\newblock {\slshape \bibinfo{journal}{The Review of Symbolic Logic}}
  \bibinfo{volume}{15}(\bibinfo{number}{1}), p. \bibinfo{pages}{165–196},
  \doi{10.1017/S1755020321000058}.

\bibitemdeclare{inproceedings}{Danecki:1985}
\bibitem{Danecki:1985}
\bibinfo{author}{R.~\surnamestart Danecki\surnameend} (\bibinfo{year}{1985}):
  \emph{\bibinfo{title}{Nondeterministic propositional dynamic logic with
  intersection is decidable}}.
\newblock In: {\slshape \bibinfo{booktitle}{Computation Theory}},
  \bibinfo{publisher}{Springer}, pp. \bibinfo{pages}{34--53},
  \doi{10.1007/3-540-16066-3_5}.

\bibitemdeclare{inproceedings}{degremontetal:2011}
\bibitem{degremontetal:2011}
\bibinfo{author}{C.~\surnamestart Degremont\surnameend},
  \bibinfo{author}{B.~\surnamestart L{\"o}we\surnameend} \&
  \bibinfo{author}{A.~\surnamestart Witzel\surnameend} (\bibinfo{year}{2011}):
  \emph{\bibinfo{title}{The synchronicity of dynamic epistemic logic}}.
\newblock In: {\slshape \bibinfo{booktitle}{Proceedings\ of 13th TARK}},
  \bibinfo{publisher}{ACM}, pp. \bibinfo{pages}{145--152},
  \doi{10.1145/2000378.2000395}.

\bibitemdeclare{article}{hvdetal.lucky:2024}
\bibitem{hvdetal.lucky:2024}
\bibinfo{author}{H.~\surnamestart van Ditmarsch\surnameend} \&
  \bibinfo{author}{M.~\surnamestart Gattinger\surnameend}
  (\bibinfo{year}{2024}): \emph{\bibinfo{title}{You can only be lucky once:
  optimal gossip for epistemic goals}}.
\newblock {\slshape \bibinfo{journal}{Mathematical Structures in Computer
  Science}}, p. \bibinfo{pages}{1–28}, \doi{10.1017/S0960129524000082}.

\bibitemdeclare{book}{hvdetal.del:2007}
\bibitem{hvdetal.del:2007}
\bibinfo{author}{H.~\surnamestart van Ditmarsch\surnameend},
  \bibinfo{author}{W.~\surnamestart van~der Hoek\surnameend} \&
  \bibinfo{author}{B.~\surnamestart Kooi\surnameend} (\bibinfo{year}{2007}):
  \emph{\bibinfo{title}{Dynamic Epistemic Logic}}.
\newblock {\slshape \bibinfo{series}{Synthese Library}} \bibinfo{volume}{337},
  \bibinfo{publisher}{Springer}, \doi{10.1007/978-1-4020-5839-4}.

\bibitemdeclare{article}{logicofgossiping:2020}
\bibitem{logicofgossiping:2020}
\bibinfo{author}{H.~\surnamestart van Ditmarsch\surnameend},
  \bibinfo{author}{W.~\surnamestart van~der Hoek\surnameend} \&
  \bibinfo{author}{L.B. \surnamestart Kuijer\surnameend}
  (\bibinfo{year}{2020}): \emph{\bibinfo{title}{The logic of gossiping}}.
\newblock {\slshape \bibinfo{journal}{Artificial Intelligence}}
  \bibinfo{volume}{286}, p. \bibinfo{pages}{103306},
  \doi{10.1016/j.artint.2020.103306}.

\bibitemdeclare{article}{FaginHV92}
\bibitem{FaginHV92}
\bibinfo{author}{R.~\surnamestart Fagin\surnameend}, \bibinfo{author}{J.Y.
  \surnamestart Halpern\surnameend} \& \bibinfo{author}{M.Y. \surnamestart
  Vardi\surnameend} (\bibinfo{year}{1992}): \emph{\bibinfo{title}{What Can
  Machines Know? On the Properties of Knowledge in Distributed Systems}}.
\newblock {\slshape \bibinfo{journal}{J. {ACM}}}
  \bibinfo{volume}{39}(\bibinfo{number}{2}), pp. \bibinfo{pages}{328--376},
  \doi{10.1145/128749.150945}.

\bibitemdeclare{unpublished}{rusbouke.aiml:2024}
\bibitem{rusbouke.aiml:2024}
\bibinfo{author}{R.~\surnamestart Galimullin\surnameend} \&
  \bibinfo{author}{L.B. \surnamestart Kuijer\surnameend}
  (\bibinfo{year}{2024}): \emph{\bibinfo{title}{Varieties of Distributed
  Knowledge}}.
\newblock \bibinfo{note}{Proceedings of the 15th {Advances in Modal Logic}
  Prague}.

\bibitemdeclare{inproceedings}{Gargov:Passy:1990}
\bibitem{Gargov:Passy:1990}
\bibinfo{author}{G.~\surnamestart Gargov\surnameend} \&
  \bibinfo{author}{S.~\surnamestart Passy\surnameend} (\bibinfo{year}{1990}):
  \emph{\bibinfo{title}{A note on Boolean Modal Logic}}.
\newblock In: {\slshape \bibinfo{booktitle}{Mathematical Logic}},
  \bibinfo{publisher}{Plenum Press}, pp. \bibinfo{pages}{299--309},
  \doi{10.1007/978-1-4613-0609-2_21}.

\bibitemdeclare{inproceedings}{Gargov:et:al:1986}
\bibitem{Gargov:et:al:1986}
\bibinfo{author}{G.~\surnamestart Gargov\surnameend},
  \bibinfo{author}{S.~\surnamestart Passy\surnameend} \&
  \bibinfo{author}{T.~\surnamestart Tinchev\surnameend} (\bibinfo{year}{1987}):
  \emph{\bibinfo{title}{Modal environment for Boolean speculations}}.
\newblock In: {\slshape \bibinfo{booktitle}{Mathematical Logic and its
  Applications}}, \bibinfo{publisher}{Plenum Press}, pp.
  \bibinfo{pages}{253--263}, \doi{10.1007/978-1-4613-0897-3_17}.

\bibitemdeclare{mastersthesis}{goldbach:2015}
\bibitem{goldbach:2015}
\bibinfo{author}{R.~\surnamestart Goldbach\surnameend} (\bibinfo{year}{2015}):
  \emph{\bibinfo{title}{Modelling Democratic Deliberation}}.
\newblock Master's thesis, \bibinfo{school}{ILLC, University of Amsterdam}.
\newblock \urlprefix\url{https://eprints.illc.uva.nl/id/document/2216/}.
\newblock \bibinfo{note}{Master of Logic Series MoL-2015-05}.

\bibitemdeclare{article}{Goranko:1996}
\bibitem{Goranko:1996}
\bibinfo{author}{V.~\surnamestart Goranko\surnameend} (\bibinfo{year}{1996}):
  \emph{\bibinfo{title}{Hierarchies of modal and temporal logics with reference
  pointers}}.
\newblock {\slshape \bibinfo{journal}{Journal of Logic, Language, and
  Information}} \bibinfo{volume}{5}, pp. \bibinfo{pages}{1--24},
  \doi{10.1007/BF00215625}.

\bibitemdeclare{article}{Goranko:Passy:1992}
\bibitem{Goranko:Passy:1992}
\bibinfo{author}{V.~\surnamestart Goranko\surnameend} \&
  \bibinfo{author}{S.~\surnamestart Passy\surnameend} (\bibinfo{year}{1992}):
  \emph{\bibinfo{title}{Using the universal modality: gains and questions}}.
\newblock {\slshape \bibinfo{journal}{Journal of Logic and Computation}}
  \bibinfo{volume}{2}, pp. \bibinfo{pages}{5--30}, \doi{10.1093/logcom/2.1.5}.

\bibitemdeclare{inproceedings}{halpernmoses:85b}
\bibitem{halpernmoses:85b}
\bibinfo{author}{J.Y. \surnamestart Halpern\surnameend} \&
  \bibinfo{author}{Y.~\surnamestart Moses\surnameend} (\bibinfo{year}{1984}):
  \emph{\bibinfo{title}{Knowledge and Common Knowledge in a Distributed
  Environment}}.
\newblock In: {\slshape \bibinfo{booktitle}{Proceedings\ of the 3rd {PODC}}},
  pp. \bibinfo{pages}{50--61}, \doi{10.1145/800222.806735}.

\bibitemdeclare{article}{halpernmoses:1990}
\bibitem{halpernmoses:1990}
\bibinfo{author}{J.Y. \surnamestart Halpern\surnameend} \&
  \bibinfo{author}{Y.~\surnamestart Moses\surnameend} (\bibinfo{year}{1990}):
  \emph{\bibinfo{title}{Knowledge and Common Knowledge in a Distributed
  Environment}}.
\newblock {\slshape \bibinfo{journal}{Journal of the {ACM}}}
  \bibinfo{volume}{37(3)}, pp. \bibinfo{pages}{549--587},
  \doi{10.1145/79147.79161}.

\bibitemdeclare{article}{Harel:1985}
\bibitem{Harel:1985}
\bibinfo{author}{D.~\surnamestart Harel\surnameend} (\bibinfo{year}{1985}):
  \emph{\bibinfo{title}{Recurring dominoes: making the highly undecidable
  highly understandable}}.
\newblock {\slshape \bibinfo{journal}{North-Holland Mathematical Studies}}
  \bibinfo{volume}{102}, pp. \bibinfo{pages}{51--71},
  \doi{10.1016/S0304-0208(08)73075-5}.

\bibitemdeclare{book}{hareletal:2000}
\bibitem{hareletal:2000}
\bibinfo{author}{D.~\surnamestart Harel\surnameend},
  \bibinfo{author}{D.~\surnamestart Kozen\surnameend} \&
  \bibinfo{author}{J.~\surnamestart Tiuryn\surnameend} (\bibinfo{year}{2000}):
  \emph{\bibinfo{title}{Dynamic Logic}}.
\newblock \bibinfo{publisher}{MIT Press}, \bibinfo{address}{Cambridge MA},
  \doi{10.7551/mitpress/2516.001.0001}.
\newblock \bibinfo{note}{Foundations of Computing Series}.

\bibitemdeclare{article}{HayekAER45}
\bibitem{HayekAER45}
\bibinfo{author}{F.~\surnamestart Hayek\surnameend} (\bibinfo{year}{1945}):
  \emph{\bibinfo{title}{The Use of Knowledge in Society}}.
\newblock {\slshape \bibinfo{journal}{American Economic Review}}
  \bibinfo{volume}{35}, pp. \bibinfo{pages}{519--530}.
\newblock \urlprefix\url{https://www.jstor.org/stable/1809376}.

\bibitemdeclare{article}{hilpinen:1977}
\bibitem{hilpinen:1977}
\bibinfo{author}{R.~\surnamestart Hilpinen\surnameend} (\bibinfo{year}{1977}):
  \emph{\bibinfo{title}{Remarks on personal and impersonal knowledge}}.
\newblock {\slshape \bibinfo{journal}{Canadian Journal of Philosophy}}
  \bibinfo{volume}{7}, pp. \bibinfo{pages}{1--9},
  \doi{10.1080/00455091.1977.10716173}.

\bibitemdeclare{article}{HoekM92}
\bibitem{HoekM92}
\bibinfo{author}{W.~\surnamestart van~der Hoek\surnameend} \&
  \bibinfo{author}{J.-J.Ch. \surnamestart Meyer\surnameend}
  (\bibinfo{year}{1992}): \emph{\bibinfo{title}{Making Some Issues of Implicit
  Knowledge Explicit}}.
\newblock {\slshape \bibinfo{journal}{Int. J. Found. Comput. Sci.}}
  \bibinfo{volume}{3}(\bibinfo{number}{2}), pp. \bibinfo{pages}{193--223},
  \doi{10.1142/S0129054192000139}.

\bibitemdeclare{inproceedings}{Lange:2005}
\bibitem{Lange:2005}
\bibinfo{author}{M.~\surnamestart Lange\surnameend} (\bibinfo{year}{2005}):
  \emph{\bibinfo{title}{A lower complexity bound for propositional dynamic
  logic with intersection}}.
\newblock In: {\slshape \bibinfo{booktitle}{Advances in Modal Logic 5}},
  \bibinfo{publisher}{College Publications}, pp. \bibinfo{pages}{133--147}.

\bibitemdeclare{article}{Lange:Lutz:2005}
\bibitem{Lange:Lutz:2005}
\bibinfo{author}{M.~\surnamestart Lange\surnameend} \&
  \bibinfo{author}{C.~\surnamestart Lutz\surnameend} (\bibinfo{year}{2005}):
  \emph{\bibinfo{title}{2-Exp Time lower bounds for propositional dynamic
  logics with intersection}}.
\newblock {\slshape \bibinfo{journal}{Journal of Symbolic Logic}}
  \bibinfo{volume}{70}(\bibinfo{number}{4}), p. \bibinfo{pages}{1072–1086},
  \doi{10.2178/jsl/1129642115}.

\bibitemdeclare{inproceedings}{Massacci:2001}
\bibitem{Massacci:2001}
\bibinfo{author}{F.~\surnamestart Massacci\surnameend} (\bibinfo{year}{2001}):
  \emph{\bibinfo{title}{Decision procedures for expressive description logics
  with intersection, composition, converse of roles and role identity}}.
\newblock In: {\slshape \bibinfo{booktitle}{Proceedings of the 17th {IJCAI}}},
  \bibinfo{publisher}{Morgan Kaufmann}, pp. \bibinfo{pages}{193--198}.

\bibitemdeclare{article}{MosesT88}
\bibitem{MosesT88}
\bibinfo{author}{Y.~\surnamestart Moses\surnameend} \& \bibinfo{author}{M.R.
  \surnamestart Tuttle\surnameend} (\bibinfo{year}{1988}):
  \emph{\bibinfo{title}{Programming Simultaneous Actions Using Common
  Knowledge}}.
\newblock {\slshape \bibinfo{journal}{Algorithmica}} \bibinfo{volume}{3}, pp.
  \bibinfo{pages}{121--169}, \doi{10.1007/BF01762112}.

\bibitemdeclare{inproceedings}{moss.handbook:2015}
\bibitem{moss.handbook:2015}
\bibinfo{author}{L.S. \surnamestart Moss\surnameend} (\bibinfo{year}{2015}):
  \emph{\bibinfo{title}{Dynamic Epistemic Logic}}.
\newblock In \bibinfo{editor}{H.~\surnamestart van Ditmarsch\surnameend},
  \bibinfo{editor}{J.Y. \surnamestart Halpern\surnameend},
  \bibinfo{editor}{W.~\surnamestart van~der Hoek\surnameend} \&
  \bibinfo{editor}{B.~\surnamestart Kooi\surnameend}, editors: {\slshape
  \bibinfo{booktitle}{Handbook of epistemic logic}},
  \bibinfo{publisher}{College Publications}, pp. \bibinfo{pages}{261--312}.

\bibitemdeclare{article}{Orlowska:1990}
\bibitem{Orlowska:1990}
\bibinfo{author}{E.~\surnamestart Or{\l}owska\surnameend}
  (\bibinfo{year}{1990}): \emph{\bibinfo{title}{Kripke semantics for knowledge
  representation logics}}.
\newblock {\slshape \bibinfo{journal}{Studia Logica}} \bibinfo{volume}{49}, pp.
  \bibinfo{pages}{255--272}, \doi{10.1007/BF00935602}.

\bibitemdeclare{inproceedings}{ParikhR85}
\bibitem{ParikhR85}
\bibinfo{author}{R.~\surnamestart Parikh\surnameend} \&
  \bibinfo{author}{R.~\surnamestart Ramanujam\surnameend}
  (\bibinfo{year}{1985}): \emph{\bibinfo{title}{Distributed Processes and the
  Logic of Knowledge}}.
\newblock In: {\slshape \bibinfo{booktitle}{Proceedings\ of Logics of
  Programs}}, {\slshape \bibinfo{series}{Lecture Notes in Computer Science}}
  \bibinfo{volume}{193}, pp. \bibinfo{pages}{256--268},
  \doi{10.1007/3-540-15648-8\_21}.

\bibitemdeclare{article}{parikhetal:2003}
\bibitem{parikhetal:2003}
\bibinfo{author}{R.~\surnamestart Parikh\surnameend} \&
  \bibinfo{author}{R.~\surnamestart Ramanujam\surnameend}
  (\bibinfo{year}{2003}): \emph{\bibinfo{title}{A knowledge based semantics of
  messages}}.
\newblock {\slshape \bibinfo{journal}{Journal of Logic, Language and
  Information}} \bibinfo{volume}{12}, pp. \bibinfo{pages}{453--467},
  \doi{10.1023/A:1025007018583}.

\bibitemdeclare{article}{Passy:Tinchev:1985}
\bibitem{Passy:Tinchev:1985}
\bibinfo{author}{S.~\surnamestart Passy\surnameend} \&
  \bibinfo{author}{T.~\surnamestart Tinchev\surnameend} (\bibinfo{year}{1985}):
  \emph{\bibinfo{title}{{PDL} with data constants}}.
\newblock {\slshape \bibinfo{journal}{Information Processing Letters}}
  \bibinfo{volume}{20}, pp. \bibinfo{pages}{35--41},
  \doi{10.1016/0020-0190(85)90127-9}.

\bibitemdeclare{article}{Passy:Tinchev:1991}
\bibitem{Passy:Tinchev:1991}
\bibinfo{author}{S.~\surnamestart Passy\surnameend} \&
  \bibinfo{author}{T.~\surnamestart Tinchev\surnameend} (\bibinfo{year}{1991}):
  \emph{\bibinfo{title}{An essay in Combinatory Dynamic Logic}}.
\newblock {\slshape \bibinfo{journal}{Information and Computation}}
  \bibinfo{volume}{93}, pp. \bibinfo{pages}{263--332},
  \doi{10.1016/0890-5401(91)90026-X}.

\bibitemdeclare{article}{plaza:2007}
\bibitem{plaza:2007}
\bibinfo{author}{J.A. \surnamestart Plaza\surnameend} (\bibinfo{year}{2007}):
  \emph{\bibinfo{title}{Logics of Public Communications}}.
\newblock {\slshape \bibinfo{journal}{Synthese}} \bibinfo{volume}{158(2)}, pp.
  \bibinfo{pages}{165--179}, \doi{10.1007/s11229-007-9168-7}.
\newblock \bibinfo{note}{Reprint of Plaza's 1989 workshop paper}.

\bibitemdeclare{article}{derijke:1992}
\bibitem{derijke:1992}
\bibinfo{author}{M.~\surnamestart de~Rijke\surnameend} (\bibinfo{year}{1992}):
  \emph{\bibinfo{title}{The modal logic of inequality}}.
\newblock {\slshape \bibinfo{journal}{Journal of Symbol Logic}}
  \bibinfo{volume}{57}, pp. \bibinfo{pages}{566--584}, \doi{10.2307/2275293}.

\bibitemdeclare{article}{swanson:1986}
\bibitem{swanson:1986}
\bibinfo{author}{D.R. \surnamestart Swanson\surnameend} (\bibinfo{year}{1986}):
  \emph{\bibinfo{title}{Undiscovered Public Knowledge}}.
\newblock {\slshape \bibinfo{journal}{The Library Quarterly: Information,
  Community, Policy}} \bibinfo{volume}{56}(\bibinfo{number}{2}), pp.
  \bibinfo{pages}{103--118}, \doi{10.1086/601720}.

\bibitemdeclare{article}{Vakarelov:1991}
\bibitem{Vakarelov:1991}
\bibinfo{author}{D.~\surnamestart Vakarelov\surnameend} (\bibinfo{year}{1991}):
  \emph{\bibinfo{title}{Modal logics for knowledge representation systems}}.
\newblock {\slshape \bibinfo{journal}{Theoretical Computer Science}}
  \bibinfo{volume}{90}, pp. \bibinfo{pages}{433--456}.

\bibitemdeclare{mastersthesis}{diego:2019}
\bibitem{diego:2019}
\bibinfo{author}{D.A. \surnamestart Vel\'azquez\surnameend}
  (\bibinfo{year}{2019}): \emph{\bibinfo{title}{Una relaci\'on entre las
  l\'ogicas modales y el enfoque topol\'ogico del c\'omputo distribuido}}.
\newblock Master's thesis, \bibinfo{school}{Instituto de Investigaciones en
  Matem\'aticas Aplicadas y en Sistemas, UNAM}, \bibinfo{address}{Mexico}.

\bibitemdeclare{phdthesis}{diego:2024}
\bibitem{diego:2024}
\bibinfo{author}{D.A. \surnamestart Vel\'azquez\surnameend}
  (\bibinfo{year}{2024}): \emph{\bibinfo{title}{Pattern Models: Dynamic
  Epistemic Logics for Distributed Systems}}.
\newblock Ph.D. thesis, \bibinfo{school}{Instituto de investigaciones en
  matem\'aticas aplicadas y en sistemas, UNAM}, \bibinfo{address}{Mexico}.

\bibitemdeclare{inproceedings}{diego:2021}
\bibitem{diego:2021}
\bibinfo{author}{D.A. \surnamestart Vel\'azquez\surnameend},
  \bibinfo{author}{A.~\surnamestart Casta\~{n}eda\surnameend} \&
  \bibinfo{author}{D.A. \surnamestart Rosenblueth\surnameend}
  (\bibinfo{year}{2021}): \emph{\bibinfo{title}{Communication Pattern Models:
  an Extension of Action Models for Dynamic-Network Distributed Systems}}.
\newblock In: {\slshape \bibinfo{booktitle}{Proceedings\ of TARK {XVIII}}},
  {\slshape \bibinfo{series}{{EPTCS}}} \bibinfo{volume}{335}, pp.
  \bibinfo{pages}{307--321}, \doi{10.4204/EPTCS.335.29}.

\end{thebibliography}

\end{document}